\documentclass[12pt, draftclsnofoot, onecolumn]{IEEEtran}

\usepackage{hyperref}

\usepackage[utf8]{inputenc} % allow utf-8 input
\usepackage[T1]{fontenc}    % use 8-bit T1 fonts
\usepackage{hyperref}       % hyperlinks
\usepackage{url}            % simple URL typesetting
\usepackage{booktabs}       % professional-quality tables
\usepackage{amsfonts}       % blackboard math symbols
\usepackage{nicefrac}       % compact symbols for 1/2, etc.
\usepackage{microtype}      % microtypography
\usepackage{lipsum}
\usepackage{amsthm}
\usepackage{amsmath}
\usepackage{amssymb}
\usepackage{delarray}
\usepackage{bm}
\usepackage{graphicx}
\usepackage{color}
\usepackage{enumitem}
\usepackage{footmisc}

\usepackage{subfig}
\usepackage{graphicx}
\usepackage{cancel}
\usepackage[export]{adjustbox}
\usepackage{comment}
\usepackage[font=small,labelfont=bf]{caption}

\newtheorem{theorem}{Theorem}

\newtheorem{lemma}{Lemma}
\newtheorem{prop}{Proposition}
\newtheorem{defin}{Definition}
\newtheorem{rem}{Remark}

\newcommand{\sfC}{\mathsf{C}}
\newcommand{\cN}{\mathcal{N}}
\newcommand{\cS}{\mathcal{S}}

\newcommand{\bz}{\mathbf{z}}

\begin{document}
\title{Best Relay Selection in Gaussian Half-Duplex Diamond Networks} 
\author{
Sarthak Jain, Soheil Mohajer, Martina Cardone
\thanks{
The authors are with the Electrical and Computer Engineering Department, University of Minnesota, Twin Cities, MN 55404 (e-mail: \{jain0122, soheil, cardo089\}@umn.edu). This research was supported by NSF under Award \#1907785.
The results in this paper were submitted in part to the 2020 IEEE International Symposium on Information Theory.
}
}
\maketitle
 
\begin{abstract}
This paper considers Gaussian half-duplex diamond $n$-relay networks, where a source communicates with a destination by hopping information through one layer of $n$ non-communicating relays that operate in half-duplex.
The main focus consists of investigating the following question: What is the contribution of a single relay on the {\em approximate capacity} of the entire network?
In particular, {\em approximate capacity} refers to a quantity that approximates the Shannon capacity within an additive gap which only depends on $n$, and is independent of the channel parameters.
This paper answers the above question by providing a fundamental bound on the {\em ratio} between the approximate capacity of the highest-performing single relay and the approximate capacity of the entire network, for any number $n$.
Surprisingly, it is shown that such a ratio guarantee is $f = 1/(2+2\cos(2\pi/(n+2)))$, that is a sinusoidal function of $n$, which decreases as $n$ increases.
It is also shown that the aforementioned ratio guarantee is {\em tight}, i.e., there exist Gaussian half-duplex diamond $n$-relay networks, where the highest-performing relay has an approximate capacity equal to an $f$ fraction of the approximate capacity of the entire network.
\end{abstract}

% keywords can be removed
%\keywords{Half Duplex \and Diamond Network \and Capacity}


\section{Introduction}
%$\varnothing \qquad a$
Relaying is foreseen to play a key role in the next generation technology,
%, which has the ambitious goal of ensuring ubiquitous communication anytime anywhere. 
%Relaying will indeed represent a vital part of this paradigm, 
promising performance enhancement of several 
%technology 
components of the evolving 5G architecture, such as %device-to-device communication~\cite{AsadiIEEEComm2014,HasanTWCOM2014}, 
vehicular communication~\cite{AktasGlobecom2016,ScheimCOMCAS2013}, millimeter wave communication~\cite{BiswasJSAC2016,YongWNET2015} and unmanned aerial vehicles communication~\cite{Mozaffari2016,Zeng2016}.
Relays can be classified into two main categories, namely {\em full-duplex} and {\em half-duplex}. While a full-duplex relay can simultaneously receive and transmit over the same time/frequency channel, a half-duplex relay has to use different times/bands for transmission and reception. 
When a node operates in full-duplex, several practical restrictions arise, among all how to properly cancel the self-interference~\cite{DuarteTVT2014,EverettTWC2016,MayankMobicom2011}.
%In addition, 
This operation might also require a significant energy consumption which cannot be sustained in scenarios where low-cost communication modules are needed and nodes have limited power supply. Given this, it is expected that half-duplex will still represent the predominant technology for next generation wireless networks~\cite{WangIEEEMag2017}.
%, as also recently announced in 3GPP Rel-13.

In wireless networks with relays, several practical challenges arise.
For instance, relays must synchronize for reception and transmission, which might result in a highly-complex process.
%, especially when the number of relays is large. 
Moreover, operating all the relays might bring to a severe power consumption, which cannot be sustained. 
With the goal of offering a suitable solution for these practical considerations, in~\cite{NazarogluTIT2014} the authors pioneered the so-called {\em wireless network simplification} problem, this problem seeks to provide fundamental guarantees on the amount of the capacity of the entire network that can be retained when only a subset of the available relays is operated.

In this paper, we investigate the network simplification problem in 
%the context of 
Gaussian half-duplex diamond $n$-relay networks, where a source communicates with a destination by hopping information through a layer of $n$ non-communicating half-duplex relays.
%that operate in {\color{blue}half-duplex.}
Our main result consists of deriving a fundamental bound on the amount of the {\em approximate capacity}\footnote{As we will thoroughly explain in Section~\ref{sec:headings}, approximate capacity refers to a quantity that approximates the Shannon capacity within an additive gap which only depends on $n$, and is independent of the channel parameters.} of the entire network that can be retained when only one relay is operated.
This bound amounts to 
$f = \frac{1}{2+2\cos(2\pi/(n+2))}$,
%$f = 1/(2+2\cos(2\pi/(n+2)))$
i.e., a fraction $f$ of the approximate capacity of the entire network can always be retained by operating a single relay. The merit of this result is to provide fundamental trade-off guarantees between network resource utilization 
%(e.g., number of operated relays) 
and network capacity. For instance, assume a Gaussian half-duplex diamond network with $n=3$ relays. Our result shows that if one wants to achieve $38 \%$ (or less) of the approximate capacity of the entire network, then it suffices to use only one relay, whereas if larger rates are desirable then it might be needed to operate two or three relays.
We also show that the guarantee $f$ is tight, i.e., there exist Gaussian half-duplex diamond $n$-relay networks where the highest-performing relay has an approximate capacity equal to $f$ of the approximate capacity of the entire network. To prove this result, we provide two network constructions (one for even values of $n$ and the other for odd values of $n$) for which this guarantee is tight.

\subsection{Related Work}
Characterizing the Shannon capacity for wireless relay networks is a long-standing open problem.
%, even for the case of $n=1$.
In recent years, several approximations for the Shannon capacity have been proposed among which the {\em constant gap} approach stands out~\cite{AvestimehrIT2011,OzgurIT2013,LimIT2011,LimISIT2014,CardoneIT2014}.
The main merit of these works is to provide an approximation that is at most an additive gap away from the Shannon capacity; this gap is only a function of the number of relays $n$, and it is independent of the values of the channel parameters; because of this property, this gap is said to be constant.
In the remaining part of the paper, we refer to such an approximation as {\em approximate capacity}.

In a half-duplex wireless network with $n$ relays, at each point on time, each relay can either receive or transmit, but not both simultaneously. 
Thus, it follows that the network can be operated in $2^n$ possible receive/transmit states, depending on the activity of each relay.
In~\cite{CardoneTIT2016}, the authors proved a surprising result: it suffices to operate any Gaussian half-duplex $n$-relay network with arbitrary topology in at most $n+1$ states (out of the $2^n$ possible ones) in order to characterize its approximate capacity.
This result generalizes the results in~\cite{BagheriISIT2010},~\cite{BrahmaISIT2012} and~\cite{BrahmaISIT2014}, which were specific to Gaussian half-duplex diamond relay networks with limited number of relays $n$.
This line of work has given rise to the following question: Can these $n+1$ states and the corresponding approximate capacity be found in polynomial time in $n$? 
%To the best of our knowledge, 
The answer to this question is open in general, 
%for Gaussian half-duplex relay networks with arbitrary topology, 
and it is known only for {\em paths}, i.e., the so-called line networks~\cite{EzzeldinISIT2017}, and for a specific class of layered networks~\cite{EtkindTIT2014}.
Recently, in~\cite{JainISIT2019}, the authors discovered sufficient conditions for Gaussian half-duplex $n$-relay diamond networks, which guarantee that the approximate capacity, as well as a corresponding set of $n+1$ optimal states, can be found in polynomial time in $n$.

In this work, we are interested in providing fundamental guarantees on the approximate capacity of the entire network that can be retained when only one relay is operated.
This problem was first formulated in~\cite{NazarogluTIT2014} for Gaussian full-duplex $n$-relay diamond networks: it was proved that there always exists a sub-network of $k\leq n$ relays that achieves at least a fraction of $k/(k+1)$ of the approximate capacity of the entire network.
Moreover, the authors showed that this bound is tight, i.e., there exist Gaussian full-duplex $n$-relay diamond networks in which the highest-performing sub-network of $k$ relays has an approximate capacity equal to $k/(k+1)$ of the approximate capacity of the entire network.
%, and that such a high-performing $k$-relay subnetwork can be found in polynomial time in $n$.
Recently, in~\cite{EzzeldinISIT2016} the authors analyzed the guarantee of selecting the highest-performing path in Gaussian full-duplex $n$-relay networks with arbitrary layered topology.
Very few results exist on the network simplification problem in half-duplex networks.
In~\cite{BrahmaISIT2014Second}, the authors showed that in any Gaussian half-duplex $n$-relay diamond network, there always exists a $2$-relay sub-network that has approximate capacity at least equal to $1/2$ of the approximate capacity of the entire network.
Recently, in~\cite{EzzeldinIT2019} the authors proved a tight guarantee for Gaussian half-duplex $n$-relay diamond network: there always exists an $(n-1)$-relay sub-network that retains at least $(n-1)/n$ of the approximate capacity of the entire network. Moreover, they showed that when $n \gg 1$, then for $k=1$ and $k=2$ this guarantee becomes $1/4$ and $1/2$, respectively, i.e., the fraction guarantee decreases as $n$ increases.
These results are fundamentally different from full-duplex~\cite{NazarogluTIT2014}, where the ratio guarantee 
%was shown to be
is independent of $n$. 
The main merit of our work is to provide an answer to a question that was left open in~\cite{EzzeldinIT2019}, namely: What is the fundamental guarantee (in terms of ratio) when $k=1$ relay is operated, as a function of $n$?

\subsection{Paper Organization.} 
Section~\ref{sec:headings} describes the Gaussian half-duplex diamond relay network, and defines its approximate capacity.
Section~\ref{sec:MainRes} presents the main result of the paper, by providing a tight bound on the approximate capacity of the best relay with respect to the  entire network approximate capacity.
Section~\ref{sec:ProofMainTh} provides the proof of the bound, and
%, which is the main result of this paper. Finally, in 
Section~\ref{sec:TightTh} presents some network realizations (for even and odd numbers of relays) that satisfy the bound with equality, hence showing that the ratio proved in Section~\ref{sec:ProofMainTh} is tight. 
%Finally, Section~\ref{sec:concl} concludes the paper.
Some of the more technical proofs are in the Appendix.

\section{Network Model}
\label{sec:headings}
\noindent {\em Notation.} 
For two integers $n_1$ and $n_2 \geq n_1$, $[n_1:n_2]$ indicates the set of integers from $n_1$ to $n_2$.
For a complex number $a$, $|a|$ denotes the magnitude of $a$.
Calligraphic letters (e.g., $\mathcal{A}$) denote sets. For two sets $\mathcal{A}$ and $\mathcal{B}$,  $\mathcal{A} \subseteq \mathcal{B}$ indicates that $\mathcal{A}$ is a subset of $\mathcal{B}$, and $\mathcal{A} \cap \mathcal{B}$ denotes the intersection between $\mathcal{A}$ and $\mathcal{B}$. The complement of a set $\mathcal{A}$ is indicated as $\mathcal{A}^c$; $\varnothing$ is the empty set. $\mathbb{E}[\cdot]$ denotes the expected value. Finally,  $\lfloor x \rfloor$ is the floor of $x$.
%, that is the largest integer not exceeding $x$.

\begin{figure}[t]
\includegraphics[width=0.42\columnwidth]{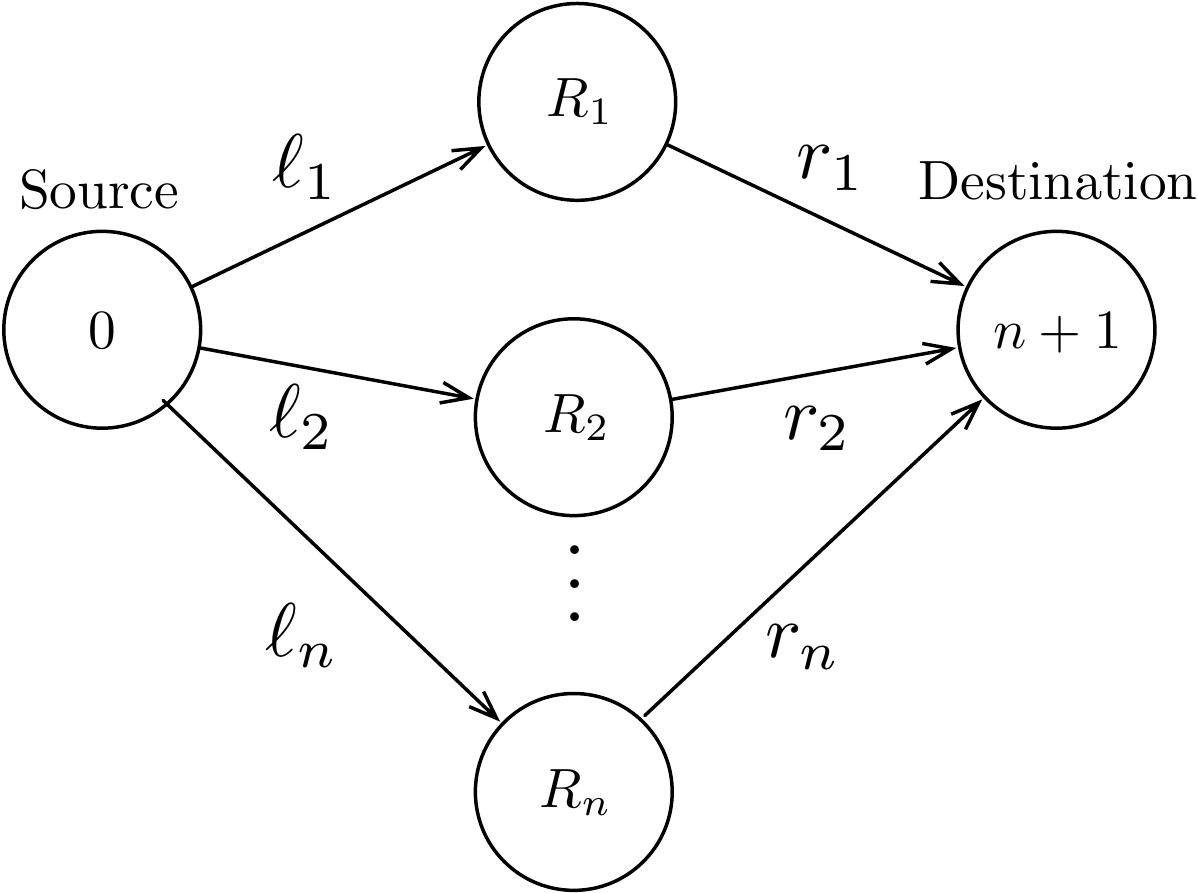}
\centering
\caption{Gaussian half-duplex diamond network with $n$ relays. }
\label{fig1}
\vspace{-5mm}
\end{figure}

The Gaussian half-duplex diamond $n$-relay network $\mathcal{N} $ consists of two hops (and three layers of nodes), as shown in Fig.~\ref{fig1}: the broadcast hop between the source (node $0$) and the set of $n$ relays $\{R_1,R_2,...,R_n\}$; and the multiple access hop between the relays $\{R_1,R_2,...,R_n\}$ and the destination (node $n+1$). The $n$ relays are assumed to be non-interfering, and the source can communicate to the destination only by hopping information through the relays, i.e., 
there is no direct link from the source to the destination. 
%Moreover, 
Relays operate in half-duplex mode, i.e., at any given time they can either receive or transmit, but not both simultaneously.  The input/output relationship for the Gaussian half-duplex diamond $n$-relay network at time $t$ is defined as
\begin{subequations}
\label{eq:Diam}
\begin{align}
Y_{i,t}&=(1-S_{i,t}) (h_{si} X_{0,t}+Z_{i,t}), \quad \forall i \in [1:n],
\\ Y_{n+1,t}&=\sum_{i=1}^{n} S_{i,t} h_{id} X_{i,t} + Z_{n+1,t},
\end{align}
\end{subequations}
where: (i) $S_{i,t}$ is a binary variable that indicates the state of relay $R_i$ at time $t$; specifically, $S_{i,t}=0$ means that relay $R_i$ is in receiving mode at time $t$, and $S_{i,t}=1$ means that relay $R_i$ is in transmitting mode at time $t$; 
(ii) $X_{i,t}, \ \forall i \in [0:n]$ is the channel input of node $i$ at time $t$ that satisfies the unit average power constraint $\mathbb{E}[|X_{i,t}|^2] \leq 1$; (iii) $h_{si}$ and $h_{id}$ are the {\em time-invariant}\footnote{The channel coefficients are assumed to remain constant for the entire transmission duration, and hence they are known to all the nodes in the network.} complex 
%fading 
channel gains  from the source to relay $R_i$ and from relay $R_i$ to the destination, respectively; (iv) $Z_{i,t} \; , i \in [1:n+1]$ is the complex additive white Gaussian noise at node $i$; noises are independent and identically distributed as $\mathcal{C} \mathcal{N}(0,1)$;
%Note that this model is without loss of generality because power constraints and/or noise variances that are non-unitary can be incorporated inside the channel coefficients.
and finally (v) $Y_{i,t}, \ \forall i \in [1:n+1]$ is the received signal by node $i$ at time instance $t$. 

The Shannon capacity (a.k.a. the maximum amount of information flow) for the Gaussian half-duplex diamond $n$-relay network in~\eqref{eq:Diam} is unknown in general, and hence its computation is notoriously an open problem (even for the case of one relay). However, it is known that the cut-set bound provides an upper bound on the channel capacity \cite{cover2012elements}. 
%On the other hand,  
Moreover, several relaying schemes, such as Quantize-Map-and-Forward (QMF)~\cite{OzgurIT2013} and Noisy Network Coding (NNC)~\cite{LimIT2011} have been shown to achieve rates that are within a {\em  constant additive gap} 
from the Shannon capacity. We continue with the following definition. 
\begin{defin} 
\label{def:ConstGap}
For the Gaussian half-duplex diamond $n$-relay network described in~\eqref{eq:Diam}, define 
%its approximate capacity as
\begin{equation} \label{maxmincut}
\begin{split}
   \mathsf{C}_n({\mathcal{N}})= \max_{\bm{\lambda}}  &  \ t \\
    \text{s.t. } & \ t \leq \sum_{\mathcal{S}\subseteq{[1:n]}} \lambda_{\mathcal{S}}\bigg( \max_{i \in \mathcal{S}^c \cap \Omega^c} \ell_i +\max_{i \in \mathcal{S} \cap \Omega} r_i \bigg),  \quad \forall   \Omega \subseteq [1:n],\\
     & \ \sum_{\mathcal{S}\subseteq{[1:n]}} \lambda_{\mathcal{S}}=1, \; \lambda_{\mathcal{S}} \geq 0, \quad \forall  \mathcal{S} \subseteq [1:n], \\
\end{split}
\end{equation}
where, $\forall i \in [1:n]$,
\begin{equation*}
    \ell_i=\log (1+|h_{si}|^2), \;\; r_i=\log (1+|h_{id}|^2).
\end{equation*}
\end{defin} 
In the above definition, $\ell_i$ and $r_i$ are the point-to-point capacities of the link from the source to relay $R_i$ and the link from relay $R_i$ to the destination, respectively. 
Moreover, in~\eqref{maxmincut}, we have that:
%As discussed earlier, $\mathcal{C}_n(\mathcal{N})$ is referred to as the \textit{Approximate Capacity} of the network $\mathcal{N}$ in this paper. Here 
(i) $\mathcal{S} \subseteq [1:n]$ corresponds to the state of the network in which the relays $R_i, i \in \mathcal{S},$ are in transmitting mode, while the rest of the relays are in receiving mode;
(ii) $\lambda_{\mathcal{S}}$ denotes the fraction of time that the network operates in state $\mathcal{S}$;
(iii) $\bm{\lambda}$ is the vector obtained by stacking together $\lambda_{\mathcal{S}}, \forall \mathcal{S} \subseteq [1:n]$, and is referred to as a {\em schedule} of the network;
(iv) $\Omega \subseteq [1:n]$ is used to denote a partition of the relays in the `side of the source', i.e., $\{0\} \cup \Omega$ is a cut of the network; 
similarly, $\Omega^c = [1:n] \setminus \Omega$  denotes a partition of the relays in the `side of the destination';
note that, for a relay $R_i, i \in \Omega,$ to contribute to the flow of information we also need $i \in \mathcal{S}$; similarly, for a relay $R_i, i \in \Omega^c,$ to contribute to the flow of information we also need $i \in \mathcal{S}^c$.

The following proposition is a consequence of~\cite{AvestimehrIT2011, OzgurIT2013,CardoneIT2014}, and shows that $\mathsf{C}_n({\mathcal{N}})$
% the approximate capacity defined 
in Definition~\ref{def:ConstGap} is within a constant additive gap from the Shannon capacity. 
Because of this property, in the remaining of the paper we refer to $\mathsf{C}_n({\mathcal{N}})$ as {\em approximate capacity}.

\begin{prop} 
Let $\mathsf{C}^G_n({\mathcal{N}})$ be the Shannon capacity of the Gaussian half-duplex diamond $n$-relay network $\mathcal{N}$ in~\eqref{eq:Diam}, and $\mathsf{C}_n({\mathcal{N}})$ be the quantity 
%is the approximate capacity 
defined in Definition~\ref{def:ConstGap}. Then, 
\begin{align*}
\left |\mathsf{C}^G_n({\mathcal{N}}) - \mathsf{C}_n({\mathcal{N}}) \right| \leq \kappa_n,
\end{align*}
where $\kappa_n$ only depends on the number of relays $n$, and is independent of the channel coefficients.
\end{prop}

%The quantity $\mathsf{C}_n({\mathcal{N}})$ is usually referred to as approximate capacity.

%An appealing feature of the approximation in Definition~\ref{def:ConstGap} is that it holds universally, i.e., independently of the channel parameters.
%With the results in~\cite{AvestimehrIT2011,OzgurIT2013,LimIT2011}, the approximate capacity for the Gaussian half-duplex diamond $n$-relay network $\mathcal{N}$ in~\eqref{eq:Diam} can be found as the solution of the following optimization problem

%Again, in the cut corresponding to the set $\Omega$, the relays $R_i, \; i \in \Omega$ are connected to the destination while the rest of the relays are connected to the source.

The optimization problem in~\eqref{maxmincut} seeks to maximize the source-destination information flow. This can be computed as the minimum flow across all the network cuts. Moreover, each relay can be scheduled for reception/transmission so as to maximize the information flow.
It therefore follows that the optimization problem in~\eqref{maxmincut} is a linear optimization problem with $O(2^n)$ constraints (corresponding to the $2^n$ network cuts $\Omega \subseteq [1:n]$), and $O(2^n)$ variables (corresponding to the $2^n$ network states $\mathcal{S} \subseteq [1:n]$). 
%$\lambda_{\mathcal{S}}, \forall \mathcal{S} \subseteq [1:n]$). 
In what follows, we illustrate this through a simple example.

\noindent
{\bf Example.}
Consider a Gaussian half-duplex diamond network with $n=2$.
Then, for this network there are $2^2=4$ possible cuts (as shown in Fig.~\ref{fig2}), each of which is a function of $2^2=4$ possible receive/transmit states (i.e., $R_1$ and $R_2$ are in receiving mode, $R_1$ and $R_2$ are in transmitting mode, one among $R_1$ and $R_2$ is in receiving mode and the other in transmitting mode).
Then, the optimization problem in~\eqref{maxmincut} will have the following constraints
\begin{eqnarray}
\begin{split}
  &\text{For } \Omega=\varnothing: &&t \leq \max(\ell_1,\ell_2) \lambda_{\varnothing}+ \ell_2   \lambda_{\{1\}}+ \ell_1   \lambda_{\{2\}}+ 0  \lambda_{\{1,2\}},   \\
   &\text{For } \Omega=\{1\}: &&t \leq \ell_2   \lambda_{\varnothing}+ (\ell_2+r_1)  \lambda_{\{1\}}+ 0 \; \lambda_{\{2\}}+ r_1 \; \lambda_{\{1,2\}},  \\
   &\text{For } \Omega=\{2\}: &&t \leq \ell_1   \lambda_{\varnothing}+ 0   \lambda_{\{1\}}+ (\ell_1+r_2) \; \lambda_{\{2\}}+ r_2 \; \lambda_{\{1,2\}}, \\
   &\text{For } \Omega=\{1,2\}: &&t \leq 0   \lambda_{\varnothing}+ r_1   \lambda_{\{1\}}+ r_2 \; \lambda_{\{2\}}+ \max(r_1,r_2) \; \lambda_{\{1,2\}}, \\
   &\text{Sum of } \bm{\lambda}: &&1= \lambda_{\varnothing}+\lambda_{\{1\}}+\lambda_{\{2\}}+\lambda_{\{1,2\}} , \\
&\text{Non-negativity of } \bm{\lambda} : && \bm{\lambda}  \geq 0. 
\end{split}
\end{eqnarray}
\begin{figure}[t] 
\includegraphics[width=0.8\columnwidth]{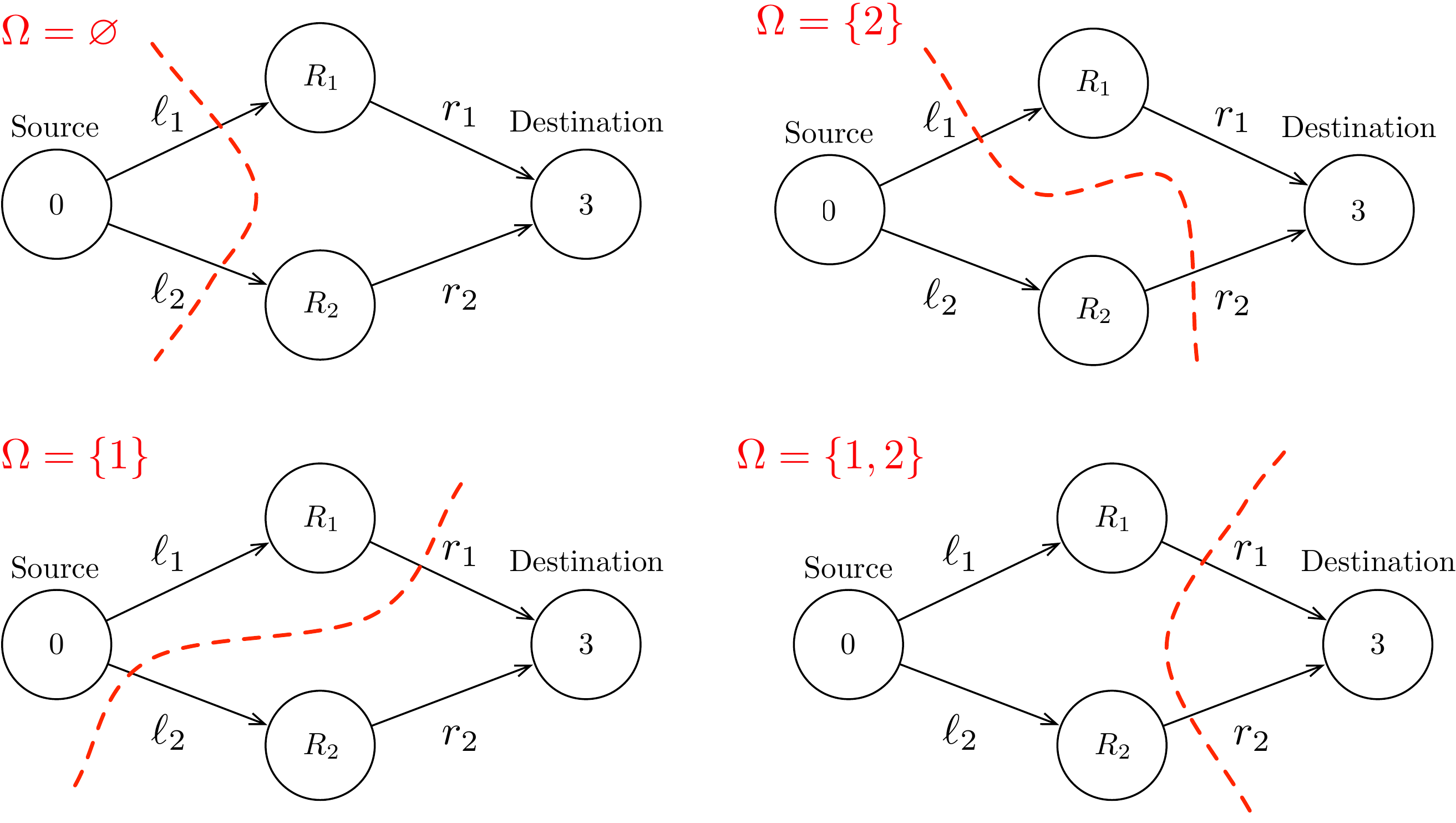}
\centering
\caption{The $4$ possible cuts in Gaussian half-duplex diamond networks with $n=2$ relays.}
\label{fig2}
\vspace{-5mm}
\end{figure}
%
%   
%
%There is a constant additive gap between the \textit{Approximate capacity} $\mathcal{C}_n(\mathcal{N})$ in (\ref{maxmincut}) and the Shannon Capacity of $\mathcal{N}$ which we denote by  $\bar{\mathcal{C}}_n(\mathcal{N})$. This constant gap is a function of the number of relays $n$ in the network, and does not depend on the channel parameters $\ell_i$ and $r_i  , \forall i \in [1:n]$. This makes it a very reasonable approximation to use. Therefore, in this work,  we will use $\mathcal{C}_n(\mathcal{N})$ to characterize the capacity of the network $\mathcal{N}$. 
%
%

\section{Problem Statement and Main Result}
\label{sec:MainRes}
An important problem in wireless communication is to characterize the fraction of the network (approximate) capacity that can be achieved by using only a subset of the relays in the network, while the remaining relays remain silent. In this paper, we address this question for a single relay case in a Gaussian half-duplex diamond $n$-relay network. 
% half-duplex relay network.  
More precisely, we characterize fundamental guarantees on the  approximate capacity of the {\em best} single relay sub-network, as a fraction of the approximate capacity of the entire network $\mathcal{N}$. 

We note that the approximate capacity $\mathsf{C}_n({\mathcal{N}})$ in~\eqref{maxmincut} is a function of the network $\mathcal{N}$ only through the point-to-point link capacities $(\ell_i,r_i), i \in [1:n]$. Thus, with a slight abuse of notation, in what follows we let $\mathcal{N} = \{(\ell_i,r_i), i \in [1:n]\}$. We also use $\mathcal{N}_i=\{(\ell_i,r_i)\}$ to denote a half-duplex network consisting of the source, relay $R_i$ and destination. 
%We formally define the {\em best} relay as the relay that has the largest approximate capacity. 
By solving the 
%optimization 
problem in~\eqref{maxmincut} for the single relay $R_i, i \in [1:n]$, we obtain that the approximate capacity of $\mathcal{N}_i$ is given~by
\begin{align*}
\mathsf{C}_{1} (\mathcal{N}_i) = \frac{\ell_i r_i}{\ell_i + r_i}.
\end{align*}
%where $\mathcal{N}_i$ denotes the network with the single relay $i$.
We also define the \emph{best single relay approximate} capacity of the network as the maximum approximate capacity among the single relay sub-networks, that is, 
%$\mathsf{C}_{1}(\mathcal{N}_i)$, and we denote its approximate capacity as
%\begin{align*}
%i^\star = \arg \max_{i \in [1:n]} \{\tilde{\mathsf{C}}_{i}\},
%\end{align*}
%and has corresponding approximate capacity of
\begin{align*}
\mathsf{C}_{1}(\mathcal{N}) = \max_{i \in [1:n]} \mathsf{C}_{1} (\mathcal{N}_i).
\end{align*}
Our goal is to find universal bounds on $\mathsf{C}_{1}(\mathcal{N})/{\mathsf{C}}_{n}(\mathcal{N})$, which holds independent of the actual value of the channel coefficients. In particular, our main result is given in the next theorem, whose proof is provided in Sections~\ref{sec:ProofMainTh}~and~\ref{sec:TightTh}.
\begin{theorem} \label{thm1}
For any Gaussian half-duplex diamond network $\mathcal{N}$ with $n$ relays and approximate capacity ${\mathsf{C}}_{n}(\mathcal{N})$, the best relay has an approximate capacity $\mathsf{C}_{1}(\mathcal{N})$ such that
\begin{equation}\label{ratio} 
 \frac{\mathsf{C}_{1}(\mathcal{N})}{{\mathsf{C}}_{n}(\mathcal{N})} \geq  \frac{1}{2+2\cos \left (\frac{2\pi}{n+2} \right )}.
\end{equation} 
Moreover, the bound in~\eqref{ratio} is tight, i.e., for any positive integer $n$, there exist Gaussian half-duplex diamond $n$-relay networks for which the best relay has an approximate capacity that satisfies the bound in~\eqref{ratio} with equality.
\end{theorem}
Fig.~\ref{fig:Ratio} provides a graphical representation of the bound in~\eqref{ratio} as a function of the number of relays $n$.
\begin{figure}[t]
\centering 
\subfloat[\label{fig:Ratio}]{
\includegraphics[width=0.47\columnwidth]{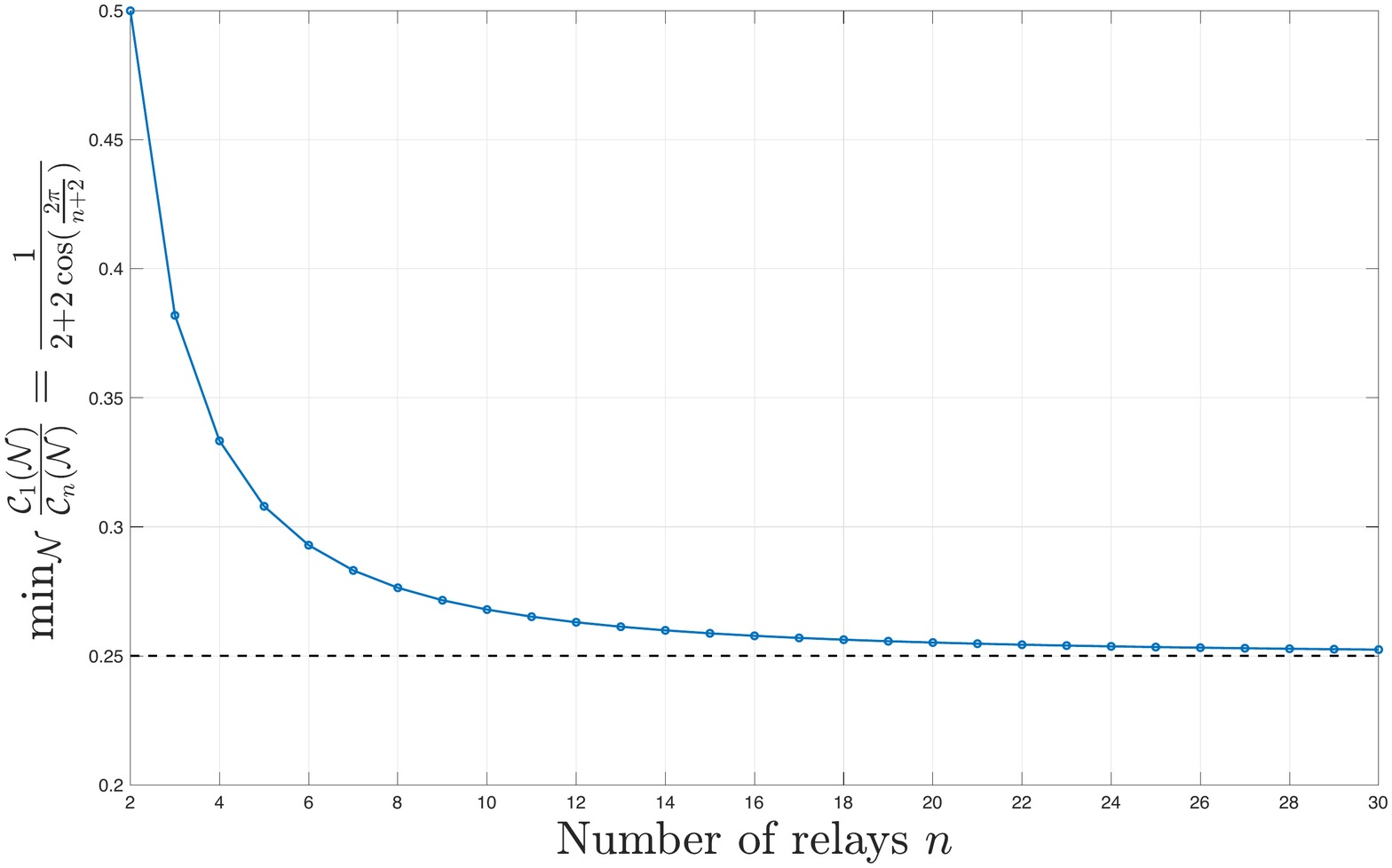}
}
\hfil
\subfloat[\label{fig:Rayleigh}]{
\includegraphics[width=0.47\columnwidth]{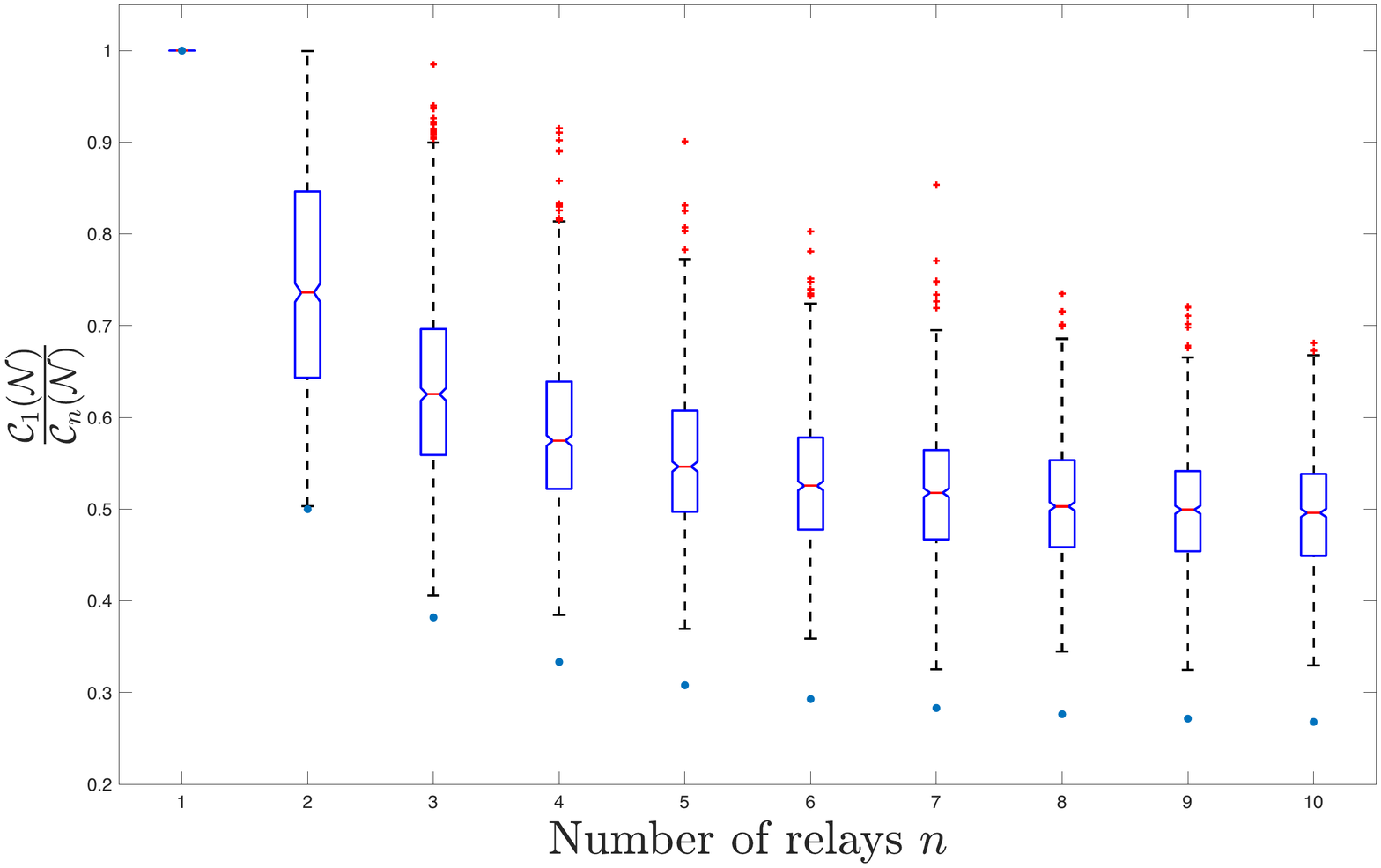}
}
\caption{Ratio  $\mathsf{C}_{1}(\mathcal{N})/{\mathsf{C}}_{n}(\mathcal{N})$ as a function of $n$: (a) the analytical bound on the ratio in~\eqref{ratio}, and (b)  the numerical ratio from $1000$ networks with random link coefficients generated from Rayleigh distribution with parameter $\sigma=1$.}
\vspace{-5mm}
\end{figure}
Before concluding this section, we state a few remarks.
\begin{rem}
The bound in~\eqref{ratio} for $n=2$ and $n \rightarrow \infty$ reduces to
\begin{align*}
\frac{\mathsf{C}_{1}(\mathcal{N})}{\mathsf{C}_{n}(\mathcal{N})} \geq \left \{
\begin{array}{ll}
1/2 & n =2, \\
1/4 & n \rightarrow \infty,
\end{array}
\right .
\end{align*}
which subsumes the result of~\cite{EzzeldinIT2019}. However,  the bound in~\eqref{ratio}  provides a tight and non-asymptotic guarantee for all values of $n$, which  was left as an open problem in~\cite{EzzeldinIT2019}.
\end{rem}
\begin{rem}
The bound in~\eqref{ratio} has a pretty surprising behavior, which depends on the cosine of a function of the number of relays $n$. This is also fundamentally different from the result in full-duplex~\cite{NazarogluTIT2014}, where it was shown that the best relay has always a capacity that is at least $1/2$ of the approximate capacity of the entire network, independent of $n$.
\end{rem}

\begin{rem}
Fig.~\ref{fig:Rayleigh} shows some of the statistics of the ratio $\mathsf{C}_{1}(\mathcal{N})/{\mathsf{C}}_{n}(\mathcal{N})$ for networks with randomly generated 
%point-to-point link capacities 
$(\ell_i,r_i), i \in [1:n],$ where $(|h_{si}|, |h_{id}|)$ follow the Rayleigh distribution with scale parameter $\sigma=1$. 
For each $n \in [1:10]$, 1000 sample networks were generated. The ratio $\mathsf{C}_{1}(\mathcal{N})/\mathsf{C}_{n}(\mathcal{N})$ for these 1000 networks is plotted as a  box-plot, wherein on each box: 
(i) the central mark indicates the median;
(ii) the top and bottom edges of the box indicate the $75^{th}$ and $25^{th}$ percentile, respectively.
Any point which is at a distance of more than 1.5 times the length of the box from the top or bottom edge is an outlier (represented by a plus sign). Whiskers are drawn from the edges of the box to the furthest observations, which are not outliers. 
The circular dots indicate the worst case ratio in~\eqref{ratio}. 
From Fig.~\ref{fig:Rayleigh}, we observe that networks with Rayleigh faded channels have a larger ratio %$C_1(\mathcal{N})/C_n(\mathcal{N})$ 
on average, compared to worst case networks. 
As an example, consider $n=3$: we have $C_1(\mathcal{N})/C_3(\mathcal{N}) \geq 0.66$ for $50\%$ of the sample networks and $C_1(\mathcal{N})/C_3(\mathcal{N}) \geq 0.72$ for $25 \%$ of the sample networks, while the worst case ratio is only $0.382$.
%Similarly for $n=10$: $C_1(\mathcal{N})/C_n(\mathcal{N})$ is greater than or equal to $0.5359$ for $50\%$ and $0.598$ for $25 \%$ of the networks, while the worse-case ratio is $0.2659$.
\end{rem}

\section{Proof of the Bound in Theorem~\ref{thm1}}
\label{sec:ProofMainTh}
In this section, we formally prove that the bound given in  Theorem~\ref{thm1} is satisfied for any Gaussian half-duplex diamond network. 
Towards this end, we first provide a few properties on the approximate capacity and  the general theory of optimization  in Section~\ref{sec:Prop}.
Then, in Section~\ref{sec:Thm1} we use these properties to prove the fraction guarantee in~\eqref{ratio}.
%Finally, in Section~\ref{sec:TightTh} we prove that the fractional guarantee in~\eqref{ratio} is indeed tight.
%In the remaining part of this paper, we define
%\begin{align}
%\label{eq:Ratio}
%g^\star = \frac{\tilde{\mathsf{C}}_{i^\star}}{{\tilde{\mathsf{C}}}_{\mathcal{N}}},
%\end{align}
%i.e., $g^\star$ represents the fraction of ${\tilde{\mathsf{C}}}_{\mathcal{N}}$ that can be guaranteed by operating only the best relay $i^\star$.

\begin{comment}
\subsection{Overview of the Proof}
We first formulate an optimization problem, whose solution provides us with the minimum ratio of $\frac{\mathsf{C}_{1}(\mathcal{N})}{{\mathsf{C}}_{n}(\mathcal{N})}$. Then, we present a few properties that any solution of the optimization problem should satisfy. This allows us to shrink the feasible set of the optimization to a smaller class of networks, in which all single relay capacities are the same. Such an optimization problem still includes many variables. Instead of directly solving it, we present a sequence of (simpler) optimization problems, which provide a sequence of upper bounds for the opiginal optimization problem. This leads to an upper bound for the ratio $\frac{\mathsf{C}_{1}(\mathcal{N})}{{\mathsf{C}}_{n}(\mathcal{N})}$, which is the claim of Theorem~\ref{thm1}. 
\end{comment}

\subsection{Properties on the Approximate Capacity}
\label{sec:Prop}
%In this section, 
Here, we derive some properties on the approximate capacity of a Gaussian half-duplex diamond $n$-relay network that we will leverage to prove the fractional guarantee in~\eqref{ratio}.
In particular, we start by stating the following three properties, which directly follow by inspection of the optimization problem in~\eqref{maxmincut}. We have,
\begin{enumerate}[label=(P\arabic*)]
    \item  \label{prop:add} The approximate capacity ${\mathsf{C}}_{n}(\mathcal{N})$ is a non-decreasing function of each point-to-point link capacity; that is, ${\mathsf{C}}_{n}(\mathcal{N} + \bm{\epsilon}) \geq {\mathsf{C}}_{n}(\mathcal{N})$, for any $2n$-vector 
%\textcolor{green}{(Is it $2$x$n$ vector?)}
$\bm{\epsilon}$ of non-negative entries. 
    
%    for  any $i\in [n]$ and $\ell'_i \geq \ell_i$ and $r'_i %\geq r_i$ we have 
%    \begin{align*}
%    {\mathsf{C}}_{n}(\{(\ell_1,\dots, \ell_{i-1},\ell'_i,\ell_{i+1}, \dots, \ell_n; r_1,\dots, r_n\}) &\geq 
%    {\mathsf{C}}_{n}(\{(\ell_1,\dots, \ell_n; r_1,\dots, r_n\})  \\
%{\mathsf{C}}_{n}(\{(\ell_1,\dots, \ell_n; r_1,\dots,r_{i-1},r'_i,r_{i+1}, \dots, r_n\})
%&\geq       {\mathsf{C}}_{n}(\{(\ell_1,\dots, \ell_n; r_1,\dots, r_n\}). 
%    \end{align*}
\item \label{prop:mult} The ratio $\mathsf{C}_{1}(\mathcal{N})/{\mathsf{C}}_{n}(\mathcal{N})$ is invariant to scaling all the point-to-point link capacities by a constant factor, that is, $\mathsf{C}_{1}(\mathcal{N})/{\mathsf{C}}_{n}(\mathcal{N}) = \mathsf{C}_{1}(\alpha \mathcal{N})/{\mathsf{C}}_{n}( \alpha \mathcal{N})$.
    \item \label{prop:label} The ratio $\mathsf{C}_{1}(\mathcal{N})/{\mathsf{C}}_{n}(\mathcal{N})$ is invariant to a relabelling of the relay nodes.
\end{enumerate}
%
%
%
%
%
%Let us suppose there exists a network $\mathcal{N}^*=\{\ell_1^*,\ell_2^*,..,\ell_n^*; r_1^*,r_2^*,..,r_n^*\}$, for which the ratio $ \frac{C_1(\mathcal{N})}{C_n(\mathcal{N})}$ is minimum, that is, $\frac{C_1(\mathcal{N})}{C_n(\mathcal{N})} \geq \frac{C_1(\mathcal{N}^*)}{C_n(\mathcal{N}^*)}$. It is easy to verify the following three properties:\\
%\begin{itemize}
%    \item The ratio $\frac{C_1(\mathcal{N})}{C_n(\mathcal{N})}$ is invariant to scaling all the channel links by a constant factor;
%    \item The network capacity $C_n(\mathcal{N})$ is a non-decreasing function of all channel gains;
%    \item The ratio $\frac{C_1(\mathcal{N})}{C_n(\mathcal{N})}$ is invariant to relabelling of the relay nodes.
%\end{itemize}
%
%Using these three properties, we have the following lemma:
%
Using the three properties above, we have the following lemma.
\begin{lemma} \label{lemma1}
Let $\mathcal{N}^\star$ be the collection of half-duplex diamond $n$-relay networks for which the ratio $\mathsf{C}_{1}(\cdot)/{\mathsf{C}}_{n}(\cdot)$ is minimum.
Then, there exists 
%a network 
$\mathcal{N} \in \mathcal{N}^\star$ that satisfies the three following properties:
\begin{subequations}
\begin{align}
& 1 \leq \ell_1 \leq \ell_2 \leq ... \leq \ell_{n-1} \leq \ell_n \leq \infty,  \label{eq:i}
\\
&  \infty \geq r_n \geq r_{n-1} \geq ... \geq r_{2} \geq r_1 \geq 1, \label{eq:ii}
\\ 
&
 \frac{\ell_i r_i}{\ell_i + r_i}=1, \quad     \forall i \in [1:n]. \label{eq:iii}
\end{align}
\end{subequations}
%
% There exists an optimal network in which the single relay capacities of all the relays $R_i, i \in [1:n]$ are equal. Furthermore, we can set all the single relay capacities equal to $1$ (without loss of generality) and hence the search for the optimal network $\mathcal{N}^*$ can be restricted to networks with the following properties:
% \begin{equation} \label{i}
%    1 \leq \ell_1 \leq \ell_2 \leq ... \leq \ell_{n-1} \leq \ell_n \leq \infty 
%\end{equation}
%
%\begin{equation} \label{ii}
%    \infty \geq r_n \geq r_{n-1} \geq ... \geq r_{2} \geq r_1 \geq 1
%\end{equation}
%
%and,
%
%\begin{equation} \label{iii}
%    \frac{\ell_i r_i}{\ell_i + r_i}=1     \forall i=1:n
%\end{equation}
\end{lemma}
\begin{proof}
We first prove that there exists $\mathcal{N} \in \mathcal{N}^\star$ for which all the $n$ single relay approximate capacities are identical. 
Consider 
%a network 
$\mathcal{N} \in \mathcal{N}^\star$ with approximate capacity ${\mathsf{C}}_{n} (\mathcal{N}) $ and ${\mathsf{C}}_{1}(\cN) = \sfC_1(\mathcal{N}_k)$, i.e., relay $R_k$ has maximum single-relay approximate capacity among all relays.   Thus, 
\begin{align}
\sfC_1(\mathcal{N}_k) = \frac{\ell_k r_k}{\ell_k+r_k} \geq \frac{\ell_j r_j}{\ell_j+r_j} = \sfC_1(\mathcal{N}_j), \quad    \forall j \in [1:n].
\end{align} 
Now we can create a new network $\cN' = \{(\ell'_i, r'_i), i\in [1:n]\}$, where
\begin{align*}
\ell'_i = \frac{\sfC_1(\cN_k)}{\sfC_1(\cN_i)} \ell_i, \qquad r'_i = \frac{\sfC_1(\cN_k)}{\sfC_1(\cN_i)}r_i, \qquad i\in [1:n].
\end{align*}
Note that since $\frac{\sfC_1(\cN_k)}{\sfC_1(\cN_i)}\geq 1$, we have $\ell'_i \geq \ell_i$ and $r'_i\geq r_i$. Hence, Property~\ref{prop:add} implies that
\begin{align}
\sfC_n(\cN') \geq \sfC_n(\cN).
\label{eq:Cn-increased}
\end{align}
Moreover, for every $i\in [1:n]$, we have 
\begin{align}
&\sfC_1(\cN'_i) = \frac{\ell'_i r'_i}{\ell'_i+r'_i} = \frac{\left(\frac{\sfC_1(\cN_k)}{\sfC_1(\cN_i)}\right)^2 \ell_i r_i}{\frac{\sfC_1(\cN_k)}{\sfC_1(\cN_i)} (\ell_i+r_i)} = \frac{\sfC_1(\cN_k)}{\sfC_1(\cN_i)} \sfC_1(\cN_i) = \sfC_1(\cN_k), \nonumber
\\
&\Longrightarrow \sfC_1(\cN') = \max_{i\in [1:n]} \sfC_1(\cN'_i) = \sfC_1(\cN_k).
\label{eq:C1-constant}
\end{align}
%. Therefore,
%\begin{align}
%\sfC_1(\cN') = \max_{i\in {\color{blue}[1:n]}} \sfC_1(\cN'_i) = \sfC_1(\cN_k).
%\label{eq:C1-constant}
%\end{align}
This together with  \eqref{eq:Cn-increased} yield to $\frac{\sfC_1(\cN')}{\sfC_n(\cN')} \leq \frac{\sfC_1(\cN)}{\sfC_n(\cN)}$,
which implies $\cN'\in \cN^\star$. Now, we can consider $\cN'' = \frac{1}{\sfC_1(\cN_k)} \cN'$. Property~\ref{prop:mult} implies that $\frac{\sfC_1(\cN'')}{\sfC_n(\cN'')} = \frac{\sfC_1(\cN')}{\sfC_n(\cN')} \leq \frac{\sfC_1(\cN)}{\sfC_n(\cN)}$, and hence $\cN'' \in \cN^\star$. Moreover, it is easy to show that in $\cN''$ we have $\sfC_1(\cN''_i) = 1$ for every $i \in [1:n]$. This proves ~\eqref{eq:iii} for the network $\cN''$.
Next, we can relabel the relay nodes such that they will be sorted in ascending order according to their left-hand link capacities $\ell''_i$, and hence satisfy~\eqref{eq:i}. Note that Property~\ref{prop:label} guarantees that the ratio $\mathsf{C}_{1}(\mathcal{N''})/{\mathsf{C}}_{n}(\mathcal{N''})$ is invariant.  Finally, combining~\eqref{eq:i} and~\eqref{eq:iii} readily proves~\eqref{eq:ii}, and concludes the proof of Lemma~\ref{lemma1}.
\end{proof}
Next, we present a lemma, that we will use in the proof of Theorem~\ref{thm1}.
\begin{lemma} \label{Lemma2}
Let $\mathcal{A}$ be any set, and $\{f_i(\cdot), i \in [1:t]\}$ be any set of functions. Then, the two optimization problems given below have identical solutions:
\begin{equation} \label{O1}
\begin{split}
\max_{\mathbf{x} \in \mathcal{A} } & \ y \\
\text{s.t. } & \ y \leq f_i(\mathbf{x}), \quad i \in [1:t],
\end{split}
    \end{equation}
% \[s.t. \;\;\;\; y \leq f_i(\mathbf{x}) \;\;\; i=1,2,...,t\]
%\[\mathbf{x}\in \mathcal{A}\]
and
\begin{equation} \label{O2}
\begin{split}
\min_{\bm{\mu}} \max_{\mathbf{x} \in \mathcal{A} } & \ \sum_{i=1}^{t} \mu_i f_i(\mathbf{x}) \\
\text{s.t. } & \ \mu_i \geq 0, \quad i \in [1:t], \\
& \ \sum_{i=1}^{t}\mu_i=1,
\end{split}
    \end{equation}
\end{lemma}

\begin{proof}
%The first optimization problem can be rewritten as:
%\[y=\max_{\mathbf{x}} (\min_{i} f_i(\mathbf{x})), \;\; \mathbf{x} \in \mathcal{A}\]
%\[=\max_{\mathbf{x}}(\min_{\mathbf{\mu}} \sum_{i=1}^{t} \mu_i f_i(\mathbf{x})),\;\; \sum_{i=1}^{t} \mu_i=1,\; \mu_i \geq 0 \;\; \forall i=1:t, \;\;\; \mathbf{x} \in \mathcal{A} \]
%\[=\max_{\mathbf{x}}(\min_{\mathbf{\mu}} \sum_{i=1}^{t} \mu_i f_i(\mathbf{x})),\;\; \sum_{i=1}^{t} \mu_i=1,\; \mu_i \geq 0 \;\; \forall i=1:t, \;\;\; \mathbf{x} \in \mathcal{A} \]
We prove Lemma~\ref{Lemma2} by showing that an optimal solution for~\eqref{O1} is a feasible solution for~\eqref{O2}, and an optimal solution for~\eqref{O2} is a feasible solution for~\eqref{O1}.

Let $\mathbf{x^\star}$ be an optimal solution for~\eqref{O1} and assume $j \in [1:t]$ be such that $f_j(\mathbf{x}^\star)\leq f_i(\mathbf{x}^\star), \forall i=[1:t]$. Then, the optimal value of~\eqref{O1} is equal to $f_j(\mathbf{x}^\star)$. Now, letting $\mu_j=1, \mu_i=0, \forall i\in [1:t], i\ne j,$ and $\mathbf{x}=\mathbf{x^\star}$ in~\eqref{O2}, we see that $f_j(\mathbf{x}^\star)$ is a feasible solution for~\eqref{O2}. 
Similarly, let $\mathbf{x}'$ be an optimal solution for~\eqref{O2} and assume $k \in [1:t]$ such that $f_k(\mathbf{x}')\leq f_i(\mathbf{x}'), \forall i=[1:t]$. 
Then, it is easy to see that the optimal $\bm{\mu}'$ in~\eqref{O2} is given by 
$\mu_k'=1, \mu_i'=0, \forall i\in [1:t], i\ne k$; moreover, the optimal value for~\eqref{O2} is equal to $f_k(\mathbf{x}')$. 
Since $\mathbf{x}' \in \mathcal{A}$ and $f_k(\mathbf{x}')\leq f_i(\mathbf{x}'), \forall i=[1:t]$, then $f_k(\mathbf{x}')$ is also a feasible solution for~\eqref{O1}.
This concludes the proof of Lemma~\ref{Lemma2}.
\end{proof}

\subsection{Proof of the Fraction Guarantee in~\eqref{ratio}}
\label{sec:Thm1}
We now use the results derived in Lemma~\ref{lemma1} and Lemma~\ref{Lemma2} to prove the ratio guarantee in~\eqref{ratio}.
We start by noting that the result in Lemma~\ref{lemma1} implies that there always exists a network $\mathcal{N}$ such that
$\mathsf{C}(\mathcal{N}_i) = 1, \forall i \in [1:n]$, and hence also $\mathsf{C}_{1} (\mathcal{N}) = 1$. Thus, proving~\eqref{ratio} reduces to proving that, for any Gaussian half-duplex diamond $n$-relay network $\mathcal{N}$ with unitary single relay approximate capacities, we always have
$ {\mathsf{C}}_{n}(\mathcal{N}) \leq  \sigma_n +2$, where $\sigma_n=2\cos(\frac{2\pi}{n+2})$,
or equivalently, 
\begin{align}
\max_{\cN: \sfC(\cN_i)=1, \forall i\in [1:n]} \mathsf{C}_{n}(\cN) \leq  \sigma_n +2.
\label{eq:opt:-1}
\end{align}
In order to rephrase the constraints in the optimization problem in \eqref{eq:opt:-1}, let us define 
\begin{align}
z_i \triangleq \ell_i -1, \qquad i\in [1:n].
\label{eq:def:z}
\end{align}
Recall that $\sfC(\cN_i) = \frac{\ell_i r_i}{\ell_i+r_i}=1$. This implies that $r_i=\frac{1}{z_i}+1$. Therefore, the class of networks of interest can be parameterized by $\bz = [z_1, z_2, \dots, z_n]$. Note that the condition in~\eqref{eq:i} implies that $0 \leq z_1 \leq z_2 \leq ... \leq z_n \leq \infty$. 
%Note also that, with this definition of $\mathbf{z}$, the network approximate capacity is a function of the network $\mathcal{N}$ only through $\mathbf{z}$, which explains why in~\eqref{eq:ObjCap} we use $\mathsf{C}_{n}(\mathbf{z})$.
Rewriting our optimization problem in~\eqref{eq:opt:-1} in terms of $z_i$'s, and using the definition of the approximate capacity in \eqref{maxmincut}, we arrive at 
\begin{equation} 
\begin{split}
   \mathrm{OPT}_0 = \max_{\bz} \max_{\bm{\lambda}} &\  \Gamma   \\
    \text{s.t. } & \ \Gamma \leq \sum_{\mathcal{S}\subseteq{[1:n]}} \lambda_{\mathcal{S}}\bigg( \max_{i \in \mathcal{S}^c \cap \Omega^c} \ell_i +\max_{i \in \mathcal{S} \cap \Omega} r_i
    %\left(1+\frac{1}{z_i}\right) 
    \bigg),  \quad \forall   \Omega \subseteq [1:n],\\
     & \ \sum_{\mathcal{S}\subseteq{[1:n]}} \lambda_{\mathcal{S}}=1, \; \lambda_{\mathcal{S}} \geq 0, \quad \forall  \mathcal{S} \subseteq [1:n], \\
     & \ell_i= 1+z_i, \quad r_i = 1+\frac{1}{z_i}, \quad i\in [1:n],\\
     & \ 0\leq z_1\leq z_2 \leq \cdots \leq z_n \leq \infty.
\end{split}
\label{eq:opt:0}
\end{equation}
 
\noindent\textbf{Reducing the Number of Constraints.} 
Note that the optimisation problem in \eqref{eq:opt:0} has one constraint each possible partition of the relays $\Omega\subseteq [1:n]$. 
Instead of considering all relay partitions, we can focus on a small class of them parameterized as $\Omega_t, \forall t \in [0:n]$, where
\begin{equation}
\Omega_t = [t+1:n], \qquad \mathrm{and} \qquad \Omega_t^c =[1:t].
\end{equation}
That is, $\Omega_t$ partitions all the relays into two groups, namely $\{t+1,\dots, n,n\}$ on the `source side', and $\{1,2,\dots, t\}$ on the `destination side'.
%divides all the nodes of the network into two groups namely $\{0,1,2,\dots, t\}$ on the source side, and $\{t+1,\dots, %n,n+1\}$ on the destination side. Recall that node $0$ is the source  and node $n+1$ is the destination.
With this, the right-hand-side of the cut constraint corresponding to $\Omega_t$ in~\eqref{eq:opt:0} %\eqref{eq:opt:-1} 
can be simplified as 
\begin{align} \label{cutset}
%\begin{split}
\sum_{\mathcal{S}\subseteq{[1:n]}} \lambda_{\mathcal{S}}&\bigg( \max_{i \in \mathcal{S}^c \cap \Omega_t^c} \ell_i +\max_{i \in \mathcal{S} \cap \Omega_t} r_i\bigg) \nonumber \\ 
%\\& \leq \sum_{S \subseteq [1:n]} \lambda_{\mathcal{S}} [\max_{i \in \mathcal{S} \cap \Omega_t} \ell_i +\max_{i \in \mathcal{S^\complement} \cap \Omega_t^\complement} r_i]\\
 & = 
\sum_{\mathcal{S}:t \in \mathcal{S}} \lambda_{\mathcal{S}} \max_{i \in \mathcal{S}^c \cap \Omega_t^c} \ell_i
+ \sum_{\mathcal{S}:t\notin \mathcal{S}} \lambda_{\mathcal{S}} \max_{i \in \mathcal{S}^c \cap \Omega_t^c} \ell_i 
+ \sum_{\mathcal{S}:t+1\in \mathcal{S}} \lambda_{\mathcal{S}} 
\max_{i \in \mathcal{S} \cap \Omega_t} r_i
+\sum_{S:t+1\notin \mathcal{S}} \lambda_{\mathcal{S}} \max_{i \in \mathcal{S} \cap \Omega_t} r_i \nonumber \\
 & \stackrel{{\rm{(a)}}}{\leq}  
  \sum_{\mathcal{S}:t \in \mathcal{S}} \lambda_{\mathcal{S}} \ell_{t-1}
 + \sum_{\mathcal{S}:t\notin \mathcal{S}} \lambda_{\mathcal{S}} \ell_t 
 +\sum_{\mathcal{S}:t+1\in \mathcal{S}} \lambda_{\mathcal{S}} r_{t+1}
 +\sum_{\mathcal{S}:t+1\notin \mathcal{S}} \lambda_{\mathcal{S}} r_{t+2} \nonumber \\
 &\stackrel{{\rm{(b)}}}{=} 
 (1-\alpha_t) \ell_{t-1} 
 +\alpha_t \ell_t 
 + (1-\alpha_{t+1})r_{t+1} 
 + \alpha_{t+1} r_{t+2} \nonumber \\
 &\stackrel{{\rm{(c)}}}{=} \bar{\alpha}_t (z_{t-1}+1) + \alpha_t (z_t+1) +  \bar{\alpha}_{t+1} \left (\frac{1}{z_{t+1}}+1 \right ) + \alpha_{t+1} \left (\frac{1}{z_{t+2}}+1 \right ) \nonumber \\
 &\triangleq g_t(\mathbf{z},\bm{\alpha}),
%\end{split}
\end{align}
where the inequality in
$\rm{(a)}$ follows from the fact that, in the first summation $t\notin \mathcal{S}^c$ implies $\mathcal{S}^c \cap \Omega_t^c \subseteq  [1:t-1]$, which together with 
$\ell_1\leq \ell_2\leq \dots \leq \ell_n$ (according to~\eqref{eq:i}) yields $\max_{i \in \mathcal{S}^c \cap \Omega_t^c} \ell_i \leq \max_{i \in  [1:t-1]} \ell_i =\ell_{t-1}$. A similar argument holds for the other three summations in $\rm{(a)}$. 
%similarly, according to~\eqref{eq:ii}, relays are ordered in descending order in their right-hand link capacities, we have $\max_{i\in \mathcal{S} \cap \Omega_t} r_i =r_{t+1}$ if $t+1 \in \mathcal{S}$ and $\max_{i\in \mathcal{S} \cap \Omega_t} r_i \leq r_{t+2}$ if $t+1 \notin \mathcal{S}$;
The equality in $\rm{(b)}$ follows by letting $\alpha_t = \sum_{\mathcal{S}:t\notin \mathcal{S}} \lambda_{\mathcal{S}}$ and $\bar{\alpha}_t=(1-\alpha_t)=\sum_{\mathcal{S}:t\in \mathcal{S}} \lambda_{\mathcal{S}}$ for $t \in [1:n]$. Finally, in $\rm{(c)}$ we replaced $\ell_t$ by $1+z_t$ and $r_t$ by $1+\frac{1}{z_t}$ for $t\in [1:n]$, according  to the constraints in \eqref{eq:opt:0}. Note that, we  define $z_i=-1$ for $i \notin [1:n]$. For instance, for $t=0$, the function $g_0(\bz, \bm{\alpha})$ reduces to 
\begin{align*}
g_0(\bz, \bm{\alpha}) = \bar{\alpha}_{1} \left (\frac{1}{z_{1}}+1 \right ) + \alpha_{1} \left (\frac{1}{z_{2}}+1 \right ).
\end{align*}
Now, by ignoring all the cut constraints except those in  $\{\Omega_t: t\in [0:n]\}$, we 
%By using the definition of $g_t(\mathbf{z},\bm{\alpha})$ in~\eqref{cutset},
obtain 
%the following optimization problem:
\begin{equation} \label{Opti}
\begin{split}
\mathrm{OPT}_1 = \max_{\mathbf{z},\bm{\alpha}} & \ \Gamma\\
 \text{s.t. }  & \  \Gamma \leq g_t(\mathbf{z},\bm{\alpha}), \quad \forall t \in [0:n], \\
 & \  \alpha_i \in [0,1],  \quad \forall i \in [0:n+1], \\
 & \ 0 \leq z_1 \leq z_2 \leq \ldots \leq z_n, \\
 & \ z_{-1}=z_0=z_{n+1}=z_{n+2}=-1.
\end{split}
\end{equation}
It is clear that $\mathrm{OPT}_0 \leq \mathrm{OPT}_1$, where $\mathrm{OPT}_0$ and $\mathrm{OPT}_1$ are the solutions of the optimization problems in~\eqref{eq:opt:0} and in~\eqref{Opti}, respectively. This follows since in~\eqref{Opti} we only considered a subset of the cut constraints that we have for solving~\eqref{eq:opt:0}, and hence we enlarged the set over which a feasible solution can be found. Moreover, variables $\alpha$'s can be uniquely determined from $\lambda$'s, but the opposite does not necessarily hold.

%The solution to (\ref{Opti}) forms an upper bound for the capacity of the network, i.e., $C_n(\mathcal{N}) \leq \mathrm{OPT}_1$, where $\mathrm{OPT}_1$ is the optimal solution of (\ref{Opti}). 

Now, using Lemma~\ref{Lemma2}, we can rewrite~\eqref{Opti}  as the following optimization problem
\begin{subequations}
\label{OPT2}
\begin{align} 
    \begin{split}
    \mathrm{OPT}_2 =    \min_{\bm{\mu}} \max_{\bm{z},\bm{\alpha}} & \ h(\bm{\mu},\bm{z},\bm{\alpha})\\
        \text{s.t. }  & \ \mu_t \geq 0 , \quad \forall  t \in [0:n], \\
         & \ \sum\nolimits_{t=0}^n \mu_t =1, \\
        & \ \alpha_i \in [0,1],  \quad \forall i \in [0:n+1], \\
        & \ 0 \leq z_1 \leq z_2 \leq \ldots \leq z_n, \\
        & \ z_{-1} = z_0 = z_{n+1}=z_{n+2}=-1,
    \end{split}
\end{align}
where
\begin{align}
\label{eq:h}
    h(\bm{\mu},\bm{z},\bm{\alpha})=\sum_{t=0}^n \mu_t g_t(\bm{z},\bm{\alpha}).
\end{align}
\end{subequations}
Therefore, by means of Lemma~\ref{Lemma2}, we have   $\mathrm{OPT}_2 = \mathrm{OPT}_1$. 
%, where $\mathrm{OPT}_1$ and $\mathrm{OPT}_2$ are the optimum solutions of~\eqref{Opti} and~\eqref{OPT2}, respectively.

\noindent\textbf{Optimum $z_t^\star$'s Are Grouped.} 
Our next step towards solving the optimization problem of interest is to show that in the optimum solution of \eqref{OPT2}, $z_t^\star$ will appear in a repeated manner, i.e., except possibly for $z_1^\star$ and $z_n^\star$, each $z_t^\star$ equals either $z_{t-1}^\star$ or $z_{t+1}^\star$. 

We start by taking the derivative of the function $h(\bm{\mu},\bm{z},\bm{\alpha})$ defined in~\eqref{eq:h} with respect to each variable $z_t$, and we obtain
\begin{align*}
&\frac{\partial}{\partial z_t}h(\bm{\mu},\bm{z},\bm{\alpha}) = (\mu_t \alpha_t + \mu_{t-1} \bar{\alpha}_{t+1})-(\mu_{t-2} \alpha_{t-1}+\mu_{t-1} \bar{\alpha}_t)\frac{1}{z_t^2},
\\& \frac{\partial^2}{\partial z_t^2}h(\bm{\mu},\bm{z},\bm{\alpha})=2(\mu_{t-2} \alpha_{t-1}+\mu_{t-1} \bar{\alpha}_t) \frac{1}{z_t^3} \geq 0.
\end{align*}
Therefore, since $\alpha_t$'s and $\mu_t$'s are non-negative variables,  $h(\bm{\mu},\bm{z},\bm{\alpha})$  is a convex function of $z_t$ for any fixed coefficient vectors $\bm{\mu}$ and $\bm{\alpha}$. Hence, at the optimum point $(\bm{\mu}^\star,\bm{z}^\star,\bm{\alpha}^\star)$ for~\eqref{OPT2}, each $z_t$ should take one of its extreme values. 
However, recall that $z_t$'s are sorted, i.e., $z_{t-1}\leq z_t \leq z_{t+1}$. This implies that 
for the optimum vector 
$\bm{z}^\star= [z_1^\star , z_2^\star ,  \cdots , z_n^\star ]$ we have\footnote{Otherwise if $z_{t-1}^\star < z_{t}^\star< z_{t+1}^\star$, the convexity of the function $h(\bm{\mu},\bm{z},\bm{\alpha})$ implies that it can be further increased by either decreasing $z_t^\star$ to $z_{t-1}^\star$ or increasing it to $z_{t+1}^\star$.} 
$z_t^\star \in \{z_{t-1}^\star,z_{t+1}^\star\}$ for $t \in [2:n-1]$. Moreover, $0\leq z_1\leq z_2$ implies  $z_1^\star \in \{0,z_2^\star\}$, and similarly, $z_{n-1}\leq z_n \leq \infty$ implies $z_n^\star \in \{z_{n-1}^\star,\infty\}$. 
More precisely, the parameters $(z_1^\star , z_2^\star , \cdots , z_n^\star)$ can be grouped into 
\begin{equation} 
\label{boundaries}
%\begin{split}
z_1^\star = \cdots = z_{t_1}^\star=\beta_1,
\
z_{t_1+1}^\star = \cdots =z_{t_2}^\star=\beta_2,
\
\ldots,
z_{t_{m-1}+1}^\star = \cdots =z_{t_m}^\star=\beta_m,
    %\end{split}
\end{equation}
%\begin{equation} 
%\label{boundaries}
%    \begin{split}
%        z_1^\star &= \ldots = z_{t_1}^\star=\beta_1, \\
%        z_{t_1+1}^\star &= \ldots =z_{t_2}^\star=\beta_2, \\
%        & \qquad  \vdots \\
%        z_{t_{m-1}+1}^\star &= \ldots =z_{t_m}^\star=\beta_m,
%    \end{split}
%\end{equation}
where $0\leq \beta_1 < \beta_2 < \cdots <\beta_{m-1} < \beta_m \leq \infty$. Note that $t_j-t_{j-1}$ (with $t_0=0$) is the number of $z_i$'s whose optimum value equals $\beta_j$. Also note that  $m$ is the number of \emph{distinct} values that the collection of $z_t^\star$'s take. Note that except for possibly $\beta_1$ and $\beta_m$, each other $\beta_j$ should be taken by at least two consecutive $z_t^\star$ and $z_{t+1}^\star$, that is $t_j-t_{j-1}\geq 2$ for $j\in [2:m-1]$. This implies that the number of distinct $\beta$'s cannot exceed $\frac{n+2}{2}$. This together with the fact that $m$ is a non-negative integer, imply  $1 \leq m \leq \lfloor \frac{n+2}{2} \rfloor$.  Moreover, if $\beta_1>0$, then $z_1^\star = z_2^\star = \beta_1$, and hence $t_1\geq 2$. Similarly, if  $z_n^\star < \infty$, we have $z_{n}^\star = z_{n-1}^\star$, and thus $t_{m}-t_{m-1}\geq 2$. In summary, we have
\begin{align}
\begin{cases}
  t_1 \geq 1 & \text{if } \beta_1=0, \\
  t_1 \geq 2 & \text{if } \beta_1>0, \\
  t_i - t_{i-1} \geq 2 & \text{for } i\in [2:m-1],\\
  t_m-t_{m-1} \geq 1 & \text{if } \beta_m=\infty, \\
  t_m-t_{m-1} \geq 2 & \text{if } \beta_m<\infty. \\
\end{cases}
\label{eq:number-beta}
\end{align}

\noindent
{\bf Example.} Consider a diamond network with $n=5$ relays. Then, for the optimum vector $\bm{z}^\star= [z_1^\star , z_2^\star ,  z_3^\star , z_4^\star , z_5^\star ]$ we have
\begin{align*}
z_1^\star \in \{0,z_2^\star\}, \quad z_2^\star \in \{z_1^\star,z_3^\star\}, \quad
z_3^\star \in \{z_2^\star,z_4^\star\}, \quad  z_4^\star \in \{z_3^\star,z_5^\star\}, \quad
z_5^\star \in \{z_4^\star, \infty \}. 
\end{align*}
%\begin{align*}
%\begin{array}{ll}
%z_1^\star \in \{0,z_2^\star\}, & z_2^\star \in \{z_1^\star,z_3^\star\}, \\ 
%z_3^\star \in \{z_2^\star,z_4^\star\}, &  z_4^\star \in \{z_3^\star,z_5^\star\}, \\
%z_5^\star \in \{z_4^\star, \infty \}. &
%\end{array}
%\end{align*}
There are several possible solutions that satisfy the conditions above. 
%In what follows we illustrate two of them.
One possibility could be
\begin{align*}
z_1^\star = z_2^\star = z_3^\star = z_4^\star = z_5^\star = \beta_1,
\end{align*}
in which case, with reference to~\eqref{boundaries}, we have $m =1$ and $t_1 = 5$.
Alternatively, we may have 
\begin{align*}
 z_1^\star = 0 = \beta_1, \qquad
 z_2^\star = z_3^\star = \beta_2,\qquad
z_4^\star = z_5^\star = \beta_3,
\end{align*}
in which case, with reference to~\eqref{boundaries}, we have $m =3$, $t_1 = 1$, $t_2 =3$ and $t_3 = 5$.
Note that, since $\beta_1 = 0$, we have $t_1=1$.
%Another possibility that one could potentially envisage is given by
%\begin{align*}
%& z_1^\star = \beta_1,
%\\
%&z_2^\star = \beta_2,
%\\
%& z_3^\star = z_4^\star = z_5^\star = \beta_3.
%\end{align*}
%However, since we have that $z_2^\star \in \{z_1^\star,z_3^\star  \}$, then we have either $\beta_2 = \beta_1$ or $\beta_2 = \beta_3$.
 \hfill$\square$

We now leverage~\eqref{boundaries} to rewrite $g_t(\bm{z},\bm{\alpha})$ in~\eqref{cutset} in terms of the optimum values of $z^\star_{t}$. In particular, we focus on functions $g_t(\bm{z},\bm{\alpha})$ for $t\in \{t_0=0,t_1,t_2,\dots, t_m=n\}$. Let $\bm{\beta} =(\beta_1,\beta_2,\dots, \beta_m)$. 
First, for $t=t_0 =0$, noting that $z_{-1}=z_0=-1$, we have 
%that $g_0(\bm{z}^\star,\bm{\alpha})$ in~\eqref{cutset} can be further upper bounded as
\begin{align}
\label{eq:C0}
  g_0(\bm{z}^\star,\bm{\alpha})&=\bar{\alpha}_1 \left (\frac{1}{z_1^\star}+1 \right ) + \alpha_1 \left (\frac{1}{z_2^\star}+1 \right ) \nonumber \\
  & \stackrel{{\rm{(a)}}}{\leq} \bar{\alpha}_1 \left (\frac{1}{z_1^\star}+1 \right ) + \alpha_1 \left (\frac{1}{z_1^\star}+1 \right ) 
  % \nonumber \\ &
   = 1+\frac{1}{\beta_1} \triangleq {G}_0(\bm{\beta}),
\end{align}
where the inequality in $\rm{(a)}$ follows from  $z_1^\star \leq z_2^\star$. Next,  for all $t \in \{t_1,t_2,\dots, t_{m-1}\}$, we obtain
  \begin{align}
\label{eq:Ci}
   g_{t_i}(\bm{z}^\star,\bm{\alpha}) & =\!
%= \alpha_t (z_t+1) + \bar{\alpha}_t (z_{t-1}+1) + \bar{\alpha}_{t+1} \left (\frac{1}{z_{t+1}}+1 \right ) + \alpha_{t+1} \left (\frac{1}{z_{t+2}}+1 \right )
 \alpha_{t_i} (z_{t_i}^\star+1) + \bar{\alpha}_{t_i} (z_{t_i-1}^\star+1) + \bar{\alpha}_{t_i+1} \left (\frac{1}{z_{t_i+1}^\star}+1 \right ) + \alpha_{t_i+1} \left (\frac{1}{z_{t_i+2}^\star}+1 \right ) \nonumber\\
&   \stackrel{{\rm{(b)}}}{\leq}\!(\alpha_{t_i}+\bar{\alpha}_{t_i})(\beta_i+1) + (\alpha_{t_i+1}+\bar{\alpha}_{t_i+1}) \left (\frac{1}{\beta_{i+1}}+1 \right ) 
% \nonumber\\ & 
\!=\!2+\beta_i + \frac{1}{\beta_{i+1}} \triangleq G_i(\bm{\beta}).
\end{align}
%(\textcolor{green}{Should (b) be an inequality? I am not sure how (b) is an equality if $t_1=1$ or $t_{m}-t_{m-1}=1$}.) 
Note that $\rm{(b)}$ follows from the fact that $t_i-t_{i-1}\geq 2$, which implies $z_{t_i-1}^\star = z_{t_i}^\star = \beta_i$, and similarly $z^\star_{t_i+1} = z_{t_i+2}^\star = \beta_{i+1}$. However, for $t_1 =1$ we have $z_0^\star = -1$, and hence $\rm{(b)}$ is an inequality, and similarly for $t_{m}-t_{m-1}=1$ we have $z^\star_{t_{m+1}} = z^\star_{n+1} = -1$ and hence $\rm{(b)}$ is also an inequality.
%% SOHEIL: Check for t_1, since it is possibe that t_1=1!!!
Finally, since $z_{n+1}=z_{n+2}=-1$ for $t=t_m=n$, we can write
\begin{align}
\label{eq:Cm}
        g_n(\bm{z}^\star,\bm{\alpha})&=\alpha_n(z_n^\star+1)+\bar{\alpha}_n (z_{n-1}^\star+1) \nonumber \\
        & \stackrel{{\rm{(c)}}}{\leq} \alpha_n(z_n^\star+1)+\bar{\alpha}_n (z_{n}^\star+1) 
% \nonumber \\ &
        =1+\beta_m \triangleq {G}_m(\bm{\beta}),
\end{align}
where the inequality in $\rm{(c)}$ holds since $z_n^\star \geq z_{n-1}^\star$. Therefore, using \eqref{eq:C0}-\eqref{eq:Cm} we can upper bound the objective function of the optimization problem in~\eqref{OPT2} as
\begin{align}
h(\bm{\mu},\bm{z},\bm{\alpha}) = \sum_{i=0}^n \mu_i g_i(\bm{z}^\star,\bm{\alpha}) &= \sum_{i\in \{t_0,\dots, t_m\}} \mu_i g_i(\bm{z}^\star,\bm{\alpha}) + \sum_{i\notin \{t_0,\dots, t_m\}} \mu_i g_i(\bm{z}^\star,\bm{\alpha}) \nonumber\\
&\leq\sum_{i=0}^{m} \mu_{t_i} G_i(\bm{\beta}) + \sum_{i\notin \{t_0,\dots, t_m\}} \mu_i g_i(\bm{z}^\star,\bm{\alpha}).
\label{eq:OPT2-cost}
\end{align}

\noindent\textbf{Further Reduction of the Constraints.} 
Recall that the optimization problem in~\eqref{OPT2} includes a minimization with respect to $\bm{\mu}$. Hence, setting more restrictions on the variable $\bm{\mu}$ can only increase the optimum cost function. Let us set  $\mu_t=0$ for $t\notin \{t_0=0, t_1, t_2, \dots, t_m=n\}$, and $\mu_{t_i} = \tilde{\mu}_i$   for $i=\{0,1,\dots, m\}$. Here $\tilde{\mu}_i$'s are arbitrary non-negative variables that sum up to $1$. Incorporating  this and the bound in~\eqref{eq:OPT2-cost} into the optimization problem in~\eqref{OPT2} leads us to 
\label{OPT3}
\begin{align} 
    \begin{split}
    \mathrm{OPT}_3 =    \min_{\tilde{\bm{\mu}}} \max_{m, \bm{\beta}} & \ \sum_{t=0}^m \tilde{\mu}_t G_t(\bm{\beta})\\
        \text{s.t. }  & \ \tilde{\mu}_t \geq 0 , \quad \forall  t \in [0:m], \\
         & \ \sum_{t=0}^m \tilde{\mu}_t =1, \\
        & \ 0 \leq \beta_1 <\beta_2< \cdots < \beta_m \leq \infty.
    \end{split}
    \label{OPT3}
\end{align}
Note that $\mathrm{OPT}_2 \leq \mathrm{OPT}_3$ since: (i) the objective function in~\eqref{OPT3} is an upper bound for that of~\eqref{OPT2}, and (ii) the feasible set for $\bm{\mu}$ in~\eqref{OPT2} is a super-set of that of $\tilde{\bm{\mu}}$ in~\eqref{OPT3}.

Finally, we can again apply Lemma~\ref{Lemma2} on the optimization problem in~\eqref{OPT3} and rewrite it as 
\begin{equation} \label{OPT4}
    \begin{split}
         \mathrm{OPT}_4 = \max_{m\in \left[1: \lfloor \frac{n+2}{2}\rfloor\right]} \max_{\bm{\beta}} & \ \Phi \\
        \text{s.t. }  &  \ \Phi \leq {G}_i(\bm{\beta}), \quad \forall i \in [0:m], \\
        & \ 0 \leq \beta_1 < \beta_2 < \cdots < \beta_m \leq \infty,
    \end{split}
\end{equation}
where $G_i(\bm{\beta})$'s are defined in~\eqref{eq:C0}-\eqref{eq:Cm}. 
Note that Lemma~\ref{Lemma2} implies that $\mathrm{OPT}_3 = \mathrm{OPT}_4$.

\noindent\textbf{Analysis of the Inner Optimization Problem.} 
Let us fix $m$ in the optimization problem in \eqref{OPT4}, and  further analyze the inner optimization problem. This yields
\begin{equation} \label{OPT5}
    \begin{split}
         \mathrm{OPT}_5(m) =  \max_{\bm{\beta}} & \ \Phi \\
        \text{s.t. }  &  \ \Phi \leq {G}_i(\bm{\beta}), \quad \forall i \in [0:m], \\
        & \ 0 \leq \beta_1 < \beta_2 < \cdots < \beta_m \leq \infty,
    \end{split}
\end{equation}
for every fixed $m\in \left[1: \lfloor \frac{n+2}{2}\rfloor\right]$.

The following lemma highlights some important properties of the optimum solution of the optimization problem defined in~\eqref{OPT5}.
\begin{lemma} 
\label{lemm:OptSol}
For every integer $m$, there exists some solution  $(\bm{\beta}^\star, \Phi^\star)$ for the optimization problem in~\eqref{OPT5} that satisfies 
\begin{align*}
G_i(\bm{\beta}^\star)  = \Phi^\star, \qquad \forall i \in [1:m-1].
\end{align*}
Moreover, if $\beta_1^\star > 0$, we have ${G}_0(\bm{\beta}^\star)=\Phi^\star$, 
and similarly, if $\beta_m^\star<\infty$, then 
${G}_m(\bm{\beta}^*)=\Phi^\star$.
\end{lemma}

\begin{proof}
We use contradiction to formally prove the claim in Lemma~\ref{lemm:OptSol}. Let $\Phi^\star$ be the optimum value of the objective function, which can be attained for each $\bm{\beta}\in\bm{B}$, where $\bm{B}$ denotes the feasible set of $\bm{\beta}$ i.e., 
\begin{align*}
\min_{i\in [0:m]} G_i(\bm{\beta}) = \Phi^\star, \qquad \forall \bm{\beta} \in \bm{B}. 
\end{align*}
If the first claim in Lemma~\ref{lemm:OptSol} does not hold, then for every $\bm{\beta}\in \bm{B}$ there exists some minimum $q(\bm{\beta})\in [1:m-1]$ such that $G_{q(\bm{\beta})}(\bm{\beta}) > \Phi^\star$, i.e., 
$G_{j}(\bm{\beta}) = \Phi^\star$  for every $j<q(\bm{\beta})$. 
 Among all optimum points $\bm{\beta}\in \bm{B}$, let $\bm{\beta}^\star$ be the one with minimum $q(\bm{\beta}^\star)$, that is, $q(\bm{\beta}) \geq q(\bm{\beta}^\star)\triangleq q$.

%$p=\arg \min_{i \in [1:m-1]} {G}_i(\bm{\beta}^\star)$, which implies $\mathrm{OPT}_5(m) = G_p(\bm{\beta}^\star)$. Assume the claim does not hold, and hence there exists some $i$  such that $G_i(\bm{\beta}^\star) > G_p(\bm{\beta}^\star)$. Assume this occurs for some indices $i>p$ (a similar argument can be aplied for the case of $i<p$). Thus there exists some $q=\min \{i>p: {G}_i(\bm{\beta}^\star)> {G}_p(\bm{\beta}^\star)\}$, i.e., $q$ is the first index after $p$ for which ${G}_q(\bm{\beta}^\star) > {G}_p(\bm{\beta}^\star)$. 

%More precisely, we assume that $q$ is the minimum integer satisfying ${G}_q(\bm{\beta}^\star) > {G}_p(\bm{\beta}^\star)$ for an optimum solution $\bm{\beta}^\star$. 

We have
\begin{align*}
    2+\beta_{q}^\star + \frac{1}{\beta_{q+1}^\star}={G}_{q }(\bm{\beta}^\star) > {G}_{q -1}(\bm{\beta}^\star)=2+\beta_{q -1}^\star + \frac{1}{\beta_{q }^\star}= \Phi^\star.
\end{align*}
It is straight-forward to see that there exists some $\hat{ \beta}_{q }$ such that $\beta_{q -1}^\star<\hat{ \beta}_{q } <\beta_{q }^\star$ and 
\begin{align*}
    2+\hat{ {\beta}}_q + \frac{1}{\beta_{q+1}^\star}=   2+\beta_{q-1}^\star + \frac{1}{\hat{ {\beta}}_q}.
\end{align*}
Thus,  for the vector $\hat{\bm{\beta}} = [\beta_1^\star , \cdots , \beta_{q-1}^\star ,\hat{\beta}_q , \beta_{q+1}^\star , \cdots , \beta_m^\star ]$  we have 
\begin{align}
\begin{split}
{G}_q( \bm{\beta}^\star) > {G}_q(\hat{\bm{\beta}}) = {G}_{q-1}(\hat{\bm{\beta}}) > {G}_{q-1}( \bm{\beta}^\star) = \Phi^\star,\\
G_j(\hat{\bm{\beta}}) = G_j(\bm{\beta}^\star) \geq \Phi^\star,\qquad j\in [0:m] \setminus \{q,q-1\}. 
\end{split} 
\label{eq:Lemma-equal-contradict}
\end{align}
Therefore $(\hat{\bm{\beta}}, \Phi^\star)$  is an optimum solution  of the optimization problem, and we have $\hat{\bm{\beta}} \in \bm{B}$. However, from \eqref{eq:Lemma-equal-contradict} we have $q(\hat{\bm{\beta}}) \leq q-1 = q(\bm{\beta}^\star) -1 $, which is in contradiction with the definition of $q=q(\bm{\beta}^\star)$ and $\bm{\beta}^\star$. 
Similarly, we can show that if $\beta_1^\star>0$ then ${G}_0(\bm{\beta}^\star)=\Phi^\star$, and if $\beta_m^\star<\infty$ then ${G}_m(\bm{\beta}^\star)=\Phi^\star$. This concludes the proof of Lemma~\ref{lemm:OptSol}.
%\textcolor{red}{INCOMPLETE!} 
%Note that $\beta_q$ does not appear in any other constraint in~\eqref{OPT5}, and hence $\Phi$ is still the optimum value of the objective function under the choice of $\bm{\beta} = \hat{\bm{\beta}}$. Hence, we have another optimum solution $\tilde{\bm{\beta}}$ for which the first index $i$ satisfying $G_i(\tilde{\bm{\beta}})>\Phi $ is $q-1$, which is less than $q$  (the first index satisfying $G_q(\bm{\beta}^\star) > \Phi$). This is in contradiction with the assumption that $q$ is the minimum index 
%Let $0<\delta<\beta_q^\star-\beta_{q-1}^\star$ be a positive number and define 
%Now let 
%{\color{magenta}$\beta_q'=\beta_q^\star-\epsilon$.} 
%Then, 
%it is not difficult to see that 
%is a feasible point for the inner optimization problem in~\eqref{OPT4Lemma}, which improves the objective function (note that $\beta_q$ does not appear in any other constraint in~\eqref{OPT4Lemma}). This is in contradiction with the initial assumption on the optimality of $\bm{\beta}^\star$. A similar argument holds for $q=\min \{i<p: {C}_i(\bm{\beta}^\star)>{C}_p(\bm{\beta}^\star)\}$. Similarly, we can show that if $\beta_1^\star>0$ then ${C}_0(\bm{\beta}^\star)=\mathrm{OPT}_4(m)$, and if $\beta_m^\star<\infty$ then ${C}_m(\bm{\beta}^\star)=\mathrm{OPT}_4(m)$ (otherwise these bounds will be loose).
%This concludes the proof of Lemma~\ref{lemm:OptSol}.
\end{proof}
We now analyze the structure of $\mathrm{OPT}_5(m)$. In particular, for a given $m$, we will find the optimum $\bm{\beta}^\star$ that satisfies Lemma~\ref{lemm:OptSol}. 
%We now show that $\mathrm{OPT}_3 = \sigma_n + 2$, where $\sigma_n=2\cos(\frac{2\pi}{n+2})$. 
%For ease of notation in what follows we let $\Phi^\star = \mathrm{OPT}_3$.
Towards this end, we distinguish the following two cases. 
\begin{enumerate}
\item[(I)] If $\beta_1^\star>0$, then we define
\begin{align}
\label{eq:B1g0}
b_0=1, \qquad  b_i=\frac{1}{\prod_{k=1}^i  \beta_k^\star}, \ \forall i \in [1:m].
\end{align}
\item[(II)] If $\beta_1^\star = 0$, then we define
\begin{align}
\label{eq:B1e0}
b_0=0, \qquad b_1=1, \qquad b_i=\frac{1}{\prod_{k=2}^i \beta_k^\star}, \ \forall i \in [2:m]. 
\end{align}
\end{enumerate}
Under both cases we have
\begin{align*}
\beta_i^\star=\frac{b_{i-1}}{b_i}, \qquad \forall i \in [1:m].
\end{align*}
Using the change of variables above and the fact that ${G}_i(\bm{\beta}^\star)=\mathrm{OPT}_5(m), i \in [1:m-1]$ (see Lemma~\ref{lemm:OptSol}), we get that
\begin{align*}
G_i(\bm{\beta}^\star) = 2+\beta_i^\star + \frac{1}{\beta_{i+1}^\star} = 2+ \frac{b_{i-1}}{b_i}+\frac{b_{i+1}}{b_i}, \qquad \forall i \in [1:m-1].
\end{align*}
Then, for a given $n$ (number of relays in the network) and $m$ (number of relays with distinct channel gains in the network), we define
\begin{equation}
\label{eq:sn}
\sigma_{n,m} \triangleq \mathrm{OPT}_5(m)-2 =  \frac{b_{i-1}}{b_i}+\frac{b_{i+1}}{b_i}, \qquad \forall i \in [1:m-1],
\end{equation}
which implies
\begin{equation}
\label{eq:LinRec}
    b_{i+1}-\sigma_{n,m} b_i +b_{i-1}=0 ,  \quad \forall i \in [1:m-1].
\end{equation}
The above expression is a linear homogeneous recurrence relation of order $2$, and hence its solution can be written as \cite{wilf2005generatingfunctionology}
\begin{equation}\label{b_i}
    b_i=uU^i+vV^i, \qquad i \in [0:m],
\end{equation}
where 
%\begin{align*}
%U  =\frac{\sigma_{n,m}+\sqrt{\sigma_{n,m}^2-4}}{2} \qquad \textrm{and} 
%\qquad
%V  =\frac{\sigma_{n,m}-\sqrt{\sigma_{n,m}^2-4}}{2}
%\end{align*}
$U$ and $V$ are the roots\footnote{The solution format in \eqref{b_i} holds only if the characteristic equation in~\eqref{eq:charfun} has simple (non-repeated) roots. Note that if $\sigma_{n,m}=2$ then we have $U=V=1$, and hence the solution of the recurrence relation would be $b_i=u+vi$. This is, however, a monotonic function of $i$, and cannot satisfy both the initial and final conditions of the recurrence relation. \label{ftnt:sigma2}
} of the characteristic equation  of the recurrence relation in~\eqref{eq:LinRec}, that is, 
\begin{align}
\label{eq:charfun}
X^2 - \sigma_{n,m} X + 1 = 0.
\end{align}
Moreover, $u$ and $v$ in \eqref{b_i} can be found from the initial conditions of the recurrence relation. In particular, under case (I) and $\beta_1^\star >0$ we have $b_0=1$ and $b_1=\frac{1}{\beta_1^\star} = G_0(\bm{\beta}^\star)-1 = \mathrm{OPT}_5(m) -1 = \sigma_{n,m}+1$. Similarly, under case (II) and $\beta_1^\star =0$ we have $b_0=0$ and $b_1=1$. 

Once $u$ and $v$ are found, we can fully express $b_i$ as a function of $\sigma_{n,m}$, for $i\in [0:m]$. Then, we can use the final condition for $\beta_m^\star$ to identify the value of $\sigma_{n,m}$. More precisely, if $\beta_m^\star=\infty$ then $b_m=0$. Otherwise, if $\beta_m^\star<\infty$, from Lemma~\ref{lemm:OptSol} we have $\sigma_{n,m}+2=\mathrm{OPT}_5(m) = {G}_m(\bm{\beta}^\star)=1+\beta_m^\star$, which implies $1+\sigma_{n,m}=\beta_m^\star=\frac{b_{m-1}}{b_m}$. 
%The value of $\sigma_{n,m}$ can be found from $\bm{\beta}^\star$. 
The optimum value of $\sigma_{n,m}$ is given in the following proposition. The proof of this proposition  can be found in Appendix~\ref{app:4cases}.
\begin{prop}
\label{lem:Smn}
The optimal value $\sigma_{n,m}$ defined in \eqref{eq:sn} is given by
%\begin{subequations}
%\label{eq:optLast}
%\begin{align}
%\label{eq:optLasta}
%\mathrm{OPT}_4(m) = \sigma_{n,m}+2,
%\end{align}
%where
\begin{align}
\begin{split}
\sigma_{n,m} =
\left \{
\begin{array}{ll}
2 \cos \left( \frac{2 \pi}{2m +2}\right ) & \text{if} \ \beta_1^\star > 0 \ \text{and} \ \beta_m^\star < \infty,
\\
2 \cos \left( \frac{2 \pi}{2m +1}\right ) & \text{if} \ \beta_1^\star > 0 \ \text{and} \ \beta_m^\star = \infty,
\\
2 \cos \left( \frac{2 \pi}{2m +1}\right ) & \text{if} \ \beta_1^\star = 0 \ \text{and} \ \beta_m^\star < \infty,
\\
2 \cos \left( \frac{2 \pi}{2m}\right ) & \text{if} \ \beta_1^\star = 0 \ \text{and} \ \beta_m^\star = \infty.
\end{array}
\right .
\end{split}
\label{eq:sigma}
\end{align}
%\end{subequations}
\end{prop}

\noindent\textbf{Optimizing Over $m$.} 
Recall from \eqref{eq:sn} that $\mathrm{OPT}_5(m) = \sigma_{n,m}+2$. Therefore, Proposition~\ref{lem:Smn} fully characterizes the optimum solution of the maximization problem in \eqref{OPT5}. 
The last step of the proof of the ratio guarantee in Theorem~\ref{thm1} consists of finding the optimal solution for the optimization problem in~\eqref{OPT4}. Recall from \eqref{OPT4} that 
\begin{align}
         \mathrm{OPT}_4 = \max_{m\in \left[1: \lfloor \frac{n+2}{2}\rfloor\right]}  \mathrm{OPT}_5(m) = 2+ \max_{m\in \left[1: \lfloor \frac{n+2}{2}\rfloor\right]} \sigma_{n,m},
\label{eq:OPT4-inner}
\end{align}
where $\sigma_{n,m}$ is given in \eqref{eq:sigma}. The following proposition provides the optimum $m$, and hence the optimum solution for the optimization problem in \eqref{OPT4}.

\begin{prop}
\label{lem:FinalOpt}
The optimal solution for the optimization problem in~\eqref{eq:OPT4-inner} is given by
\begin{align*}
\mathrm{OPT}_4 =2+ 2 \cos \left( \frac{2 \pi}{n+2}\right ).
\end{align*}
\end{prop}
\begin{proof}
In order to find the optimal solution $\mathrm{OPT}_4$ for the optimization problem in~\eqref{eq:OPT4-inner}, we need to  compute the maximum value of $\sigma_{n,m}$ over $m$ for the four different cases in Proposition~\ref{lem:Smn}. Note that all the four expressions in Proposition~\ref{lem:Smn} are increasing functions of $m$. Hence, we only need to find the maximum possible value of $m$ in each case. We can analyze the following four cases, separately. 
\begin{enumerate}
\item $\beta_1^\star > 0$ and $\beta_m^\star < \infty$. For this case, from~\eqref{eq:number-beta} we have $t_1\geq 2$ and $t_i-t_{i-1}\geq 2$ for $i \in [2:m]$.
Thus, since $t_m=n$, we get
\begin{equation*}
    n=t_m=\sum_{i=2}^m (t_i-t_{i-1})+t_1 \geq 2(m-1)+2=2m,
\end{equation*}
which implies $m \leq \frac{n}{2}$, and hence
\begin{equation*}
    \mathrm{OPT}_4= 2+ \max_{m\leq \frac{n}{2}} \sigma_{n,m} =2+ \max_{m\leq \frac{n}{2}} 2 \cos \left (\frac{2\pi}{2m+2} \right )  =2+ 2 \cos \left (\frac{2\pi}{n+2} \right).
\end{equation*}
\item $\beta_1^\star > 0$ and $\beta_m^\star = \infty$. For this case, from~\eqref{eq:number-beta}  we obtain
$t_1\geq 2$,  $t_m-t_{m-1}\geq 1$ and $t_i-t_{i-1}\geq 2$ for  $i \in [2:m-1]$. Therefore,
\begin{equation*}
    n=t_m=(t_m-t_{m-1})+\sum_{i=2}^{m-1} (t_i-t_{i-1})+t_1 \geq 1+2(m-2)+2=2m-1,
\end{equation*}
which implies $m\leq\frac{n+1}{2}$. Therefore,
\begin{equation*}
   \mathrm{OPT}_4=2+\max_{m\leq \frac{n+1}{2}} \sigma_{n,m} = 2+ \max_{m\leq \frac{n+1}{2}} 2\cos{\left(\frac{2\pi }{2m+1}\right)} =2+2\cos{\left(\frac{2\pi }{n+2}\right)}. 
\end{equation*}

\item $\beta_1^\star = 0$ and $\beta_m^\star < \infty$. For this case, from~\eqref{eq:number-beta}  we have 
$t_1 \geq 1$ and $t_i-t_{i-1}\geq 2$ for $i \in [2:m]$. Thus,
\begin{equation*}
    n=t_m=(t_m-t_{m-1})+\sum_{i=2}^{m-1}(t_i-t_{i-1})+t_1 \geq 2(m-1)+1=2m-1,
\end{equation*}
which implies $m\leq \frac{n+1}{2}$. Therefore, we obtain
\begin{equation*}
    \mathrm{OPT}_4= 2+ \max_{m\leq\frac{n+1}{2}} \sigma_{n,m} =2+ \max_{m\leq\frac{n+1}{2}} 2 \cos{\left(\frac{2\pi}{2m+1}\right)}  = 2+ 2 \cos{\left(\frac{2\pi}{n+2}\right)}.
\end{equation*}
\item $\beta_1^\star = 0$ and $\beta_m^\star = \infty$. Finally, for this case, from~\eqref{eq:number-beta}  we can write
$t_1\geq 1$,  $t_m-t_{m-1}\geq 1$ and $t_i-t_{i-1}\geq 2$ for $i \in [2:m-1]$.
Hence,
\begin{equation*}
    n=t_m=(t_m-t_{m-1})+\sum_{i=2}^{m-1}(t_i-t_{i-1})+t_1 \geq 1 +2(m-2)+1=2m-2 ,
\end{equation*}
which implies $m\leq\frac{n+2}{2}$. Therefore, we obtain
\begin{equation*}
    \mathrm{OPT}_4= 2+ \max_{m \leq \frac{n+2}{2}} \sigma_{n,m}=2+\max_{m \leq \frac{n+2}{2}} 2 \cos{\left(\frac{2\pi}{2m}\right)}= 2 +2 \cos{\left(\frac{2\pi}{n+2}\right)}.
\end{equation*}
\end{enumerate}
Therefore, for all four cases we obtain $\mathrm{OPT}_4 = 2+ 2\cos{\left(\frac{2\pi}{n+2}\right)}$, which proves our claim in Proposition~\ref{lem:FinalOpt}.
 This concludes the proof of Proposition~\ref{lem:FinalOpt}.
\end{proof}

In summary, by collecting all the results above together, we have proved that for any Gaussian half-duplex diamond $n$-relay network $\mathcal{N}$ we always have
\begin{align}
\label{eq:SummBounds}
{\mathsf{C}}_{n}(\mathcal{N}) = \mathrm{OPT}_0 \leq \mathrm{OPT}_1 = \mathrm{OPT}_2 \leq \mathrm{OPT}_3 = \mathrm{OPT}_4 = 2 + 2 \cos{\left(\frac{2\pi}{n+2}\right)},
\end{align}
where $\mathrm{OPT}_0$, $\mathrm{OPT}_1$, $\mathrm{OPT}_2$ and $\mathrm{OPT}_3$ are the optimal solutions of the optimization problems in~\eqref{eq:opt:0}, in~\eqref{Opti}, in~\eqref{OPT2} and in~\eqref{OPT4}, respectively.
This proves the inequality in~\eqref{eq:opt:-1}, and hence concludes the proof of the ratio guarantee in Theorem~\ref{thm1}.

\section{The Worst Networks: Proof of the Tightness of Theorem~\ref{thm1}}
\label{sec:TightTh}
We here prove that the bound in~\eqref{ratio} is tight, that is, for any number of relays, there exists some networks for which $\sfC(\cN_1)/\sfC_n(\cN) = 1/(2+2\cos(2\pi/(n+2)))$. 
Towards this end, for every integer $n$ we provide some constructions of  half-duplex diamond $n$-relay networks for which the best relay has an approximate capacity that satisfies the bound in~\eqref{ratio} with equality.

Our constructions are inspired by the discussion and results in Section~\ref{sec:Thm1}. More precisely, we need to satisfy all
% our optimization problems 
the bounds in~\eqref{eq:SummBounds} with equality.  %We next consider two separate cases, which depend on the value of $n$.

\noindent {\textbf{Case A.1: }} Let $n=2k$ be an even integer, and consider a half-duplex diamond $n$-relay network $\mathcal{N}$ with
\begin{align}
\label{eq:WN1}
\begin{split}
\ell_{2i} =\ell_{2i-1}&= 
%\frac{\cos \left ( (i-1) \theta \right ) - \cos \left ((i+1) \theta \right )}{\cos \left ( i \theta \right ) - \cos \left ( (i+1) \theta \right )}= 
\frac{2\sin(\theta) \sin \left ( i \theta \right ) }{\cos \left ( i \theta \right ) - \cos \left ( (i+1) \theta \right )}, \qquad i \in [1:k],
\\
r_{2i} =r_{2i-1} &= 
%\frac{\cos \left ( (i-1) \theta \right ) - \cos \left ((i+1) \theta \right )}{\cos \left ( (i-1) \theta \right )-\cos \left ( i \theta \right ) } = 
\frac{2\sin(\theta) \sin \left ( i \theta \right )}{\cos \left ( (i-1) \theta \right )-\cos \left ( i \theta \right ) }, \qquad i \in [1:k],
\\
\theta & = \frac{2 \pi}{n+2}.
\end{split}
\end{align}
It is not difficult to see that, for the network in~\eqref{eq:WN1}, 
we have that $\ell_1 \leq \ell_2 \leq \ldots \leq \ell_n$,
$r_1 \geq r_2 \geq \ldots \geq r_n$. Moreover, for every relay $t \in [1:n]$ with $i=\lfloor \frac{t+1}{2}\rfloor$, we have 
\begin{align*}
\mathsf{C}_1(\mathcal{N}_t) = \frac{\ell_t r_t}{\ell_t +
 r_t} = \left(\frac{1}{\ell_t} +\frac{1}{r_t}\right)^{-1} = \frac{2\sin(\theta) \sin \left ( i \theta \right )}{\cos \left ( (i-1) \theta \right )-\cos \left ( (i+1) \theta \right ) } =1,
\end{align*}
which implies 
\begin{align}
\label{eq:sinRelWN1}
\mathsf{C}_1(\mathcal{N})=1,
\end{align}
that is, the best relay in $\mathcal{N}$ has an approximate capacity of $1$. Finally, for every $t \in [0:n-1]$,
%$t\in [1:n-2]$ \textcolor{green}{(We could also say $t \in [1:n]$ if we assume $r_{n+1}=r_{n+2}=0$)} 
with 
%$i=\lfloor \frac{t+1}{2}\rceil$ \textcolor{green}{(
$i=\lfloor \frac{t+1}{2}\rfloor$
%?)} we have 
\begin{align}
\ell_t + r_{t+2} &= \frac{2\sin(\theta) \sin \left ( i \theta \right ) }{\cos \left ( i \theta \right ) - \cos \left ( (i+1) \theta \right )} + \frac{2\sin(\theta) \sin \left ( (i+1) \theta \right )}{\cos \left ( i \theta \right )-\cos \left ( (i+1) \theta \right ) }\nonumber\\
&= 2\sin(\theta) \frac{ 2\sin \left ( \frac{(2i+1)\theta}{2} \right ) \cos\left(\frac{\theta}{2} \right) }{2\sin \left ( \frac{(2i+1)\theta}{2}\right ) \sin\left(\frac{\theta}{2}\right)}\nonumber\\
&= 4\cos^2 \left(\frac{\theta}{2}\right) = 2\cos(\theta)+2,
\label{eq:sum-L-R}
\end{align}
where we let $\ell_0 = r_{n+1}=0$.

Consider now a two-state schedule given by 
\begin{align*}
\lambda_{\cS} = \left\{
\begin{array}{ll}
\frac{1}{2} & \textrm{if $\cS = \cS_o = \{1,3,5,\dots, 2k-1\}$},\\
\frac{1}{2} & \textrm{if $\cS = \cS_e = \{2,4,6,\dots, 2k\}$},\\
0 & \textrm{otherwise.}
\end{array}
\right.
\end{align*}
The rate $\mathsf{R}_n(\mathcal{N}) $ achieved by  this two-state schedule can be found from \eqref{maxmincut}, and satisfies 
\begin{align}
\mathsf{R}_n(\mathcal{N}) &= \min_{\Omega \subseteq [1:n]} \sum_{\mathcal{S}\subseteq{[1:n]}} \lambda_{\mathcal{S}}\bigg( \max_{i \in \mathcal{S}^c \cap \Omega^c} \ell_i +\max_{i \in \mathcal{S} \cap \Omega} r_i \bigg)\nonumber\\
&= \min_{\Omega \subseteq [1:n]} \left\{\frac{1}{2}\left( \max_{i\in \cS_e \cap \Omega^c} \ell_i + \max_{i\in \cS_o \cap \Omega} r_i\right) 
+ \frac{1}{2}\left( \max_{i\in \cS_o \cap \Omega^c} \ell_i + \max_{i\in \cS_e \cap \Omega} r_i\right)\right\} \nonumber\\
%&= \min_{\Omega \subseteq [1:n]} \left\{ \frac{1}{2}\left( \max_{i\in \cS_e \cap \Omega^c} \ell_i + \max_{i\in \cS_e \cap \Omega} r_i\right) 
%+ \frac{1}{2}\left( \max_{i\in \cS_o \cap \Omega^c} \ell_i + \max_{i\in \cS_o \cap \Omega} r_i \right)\right\}\nonumber\\
&\stackrel{\rm{(a)}}{=} \min_{\Omega \subseteq [1:n]} \left\{ \frac{1}{2}\left(  \ell_t + \max_{i\in \cS_e \cap \Omega} r_i\right) 
+ \frac{1}{2}\left(   \ell_s + \max_{i\in \cS_o \cap \Omega} r_i \right)\right\}\nonumber\\
&\stackrel{\rm{(b)}}{\geq} \min_{\Omega \subseteq [1:n]} \left\{ \frac{1}{2}\left(  \ell_t + r_{t+2}\right) 
+ \frac{1}{2}\left(   \ell_s +   r_{s+2} \right)\right\}\nonumber\\
&\stackrel{\rm{(c)}}{=} \min_{\Omega \subseteq [1:n]} \left\{ \frac{1}{2}\left(  2\cos(\theta)+2\right) 
+ \frac{1}{2}\left(    2\cos(\theta)+2 \right)\right\}\nonumber\\
&=  2\cos(\theta)+2,
\label{eq:CLB}
\end{align} 
where in $\rm{(a)}$ we set $t=\max \cS_e \cap \Omega^c$ and $s=\max \cS_o \cap \Omega^c$, and $\rm{(b)}$ is due to the fact that if $t=\max \cS_e \cap \Omega^c$ then $t+2$ is an even number that belongs to $\Omega$, and similarly $s+2\in \cS_o \cap \Omega$. Finally in $\rm{(c)}$ we used the equality derived in \eqref{eq:sum-L-R}. Therefore, the rate of $ 2\cos(\theta)+2$ is achievable for this network. Moreover, note that the approximate capacity $\mathsf{C}_n(\mathcal{N})$ of a Gaussian half-duplex diamond $n$-relay network is always upper bounded by that of the same network when operated in full-duplex mode (i.e., each relay can transmit and receive simultaneously). Also, note that, for the network in~\eqref{eq:WN1}, we have that $r_1 = \max_{i \in [1:n]} r_i$. Hence, we have 
\begin{align}
\label{eq:CUB}
\mathsf{C}_n(\mathcal{N}) \leq \mathsf{C}_n^{{\rm{FD}}}(\mathcal{N}) \leq r_1= \frac{2\sin^2(\theta)}{1-\cos(\theta)} = 2\cos(\theta)+2. 
\end{align}
Finally, \eqref{eq:CLB} and \eqref{eq:CUB} imply $\mathsf{C}_n(\mathcal{N}) =  2\cos(\theta)+2$. This together with \eqref{eq:sinRelWN1} leads to
\begin{align} 
  \frac{\mathsf{C}_{1}(\mathcal{N})}{{\mathsf{C}}_{n}(\mathcal{N})} = \frac{1}{2\cos \left (\theta \right )+2} =  \frac{1}{2\cos \left (\frac{2\pi}{n+2} \right )+2}
\end{align}
for the network defined in~\eqref{eq:WN1}, and hence proves the tightness of the bound in~\eqref{ratio} when $n$ is even. Note that this network corresponds to Case~I of the network analysis in Appendix~\ref{app:4cases}, where $\beta_1^\star>0$ and $\beta_m^\star<\infty$. 
An example of the network construction in~\eqref{eq:WN1} for $n=6$ is provided in Fig.~\ref{fig:even-1}.

\noindent {\textbf{Case A.2: }} There is also another network for even values of $n=2k$ that achieves the bound in~\eqref{ratio}. This network is given by 
\begin{align}
\label{eq:WN1-2}
\begin{split}
\ell_1&=r_n=1,\qquad 
r_1 =\ell_n = L \rightarrow \infty, \\
\ell_{2i} &=\ell_{2i+1}= 
\frac{ \sin \left ( i \theta \right ) + \sin \left ( (i+1)\theta \right ) }{\sin \left ( (i+1) \theta \right )}, \qquad i \in [1:k-1],
\\
r_{2i} &=r_{2i+1} = 
\frac{ \sin \left ( i \theta \right ) + \sin \left ( (i+1)\theta \right ) }{\sin \left ( i \theta \right )}, \qquad i \in [1:k-1],
\\
\theta & = \frac{2 \pi}{n+2}.
\end{split}
\end{align}It is easy to check that for this network we also have $\mathsf{C}_1(\mathcal{N})=1$ and $\mathsf{C}_n(\mathcal{N}) = 2\cos (\theta) +2$, which can be achieved using the two-state schedule
\begin{align*}
\lambda_{\cS} = \left\{
\begin{array}{ll}
\frac{1}{2} & \textrm{if $\cS = \cS_o = \{3,5,\dots, 2k-1,2k\}$},\\
\frac{1}{2} & \textrm{if $\cS = \cS_e = \{2,4,6,\dots, 2k\}$},\\
0 & \textrm{otherwise.}
\end{array}
\right.
\end{align*}
Note that in this schedule relay $R_1$ is (asymptotically) always in receive mode and relay $R_n$ is always in transmit mode. This leads to 
\begin{align*} 
  \frac{\mathsf{C}_{1}(\mathcal{N})}{{\mathsf{C}}_{n}(\mathcal{N})} =  \frac{1}{\cos \left (\frac{2\pi}{n+2} \right )+2}.
\end{align*}
Note that this network corresponds to Case IV of the network analysis in Appendix~\ref{app:4cases}, where $\beta_1^\star=0$ and $\beta_m^\star=\infty$. The realization of this network configuration for  $n=6$ is provided in Fig.~\ref{fig:even-2}. 

\begin{figure}
	\centering
	\subfloat[]{\adjustbox{valign=c}{\includegraphics[width = 0.35\textwidth]{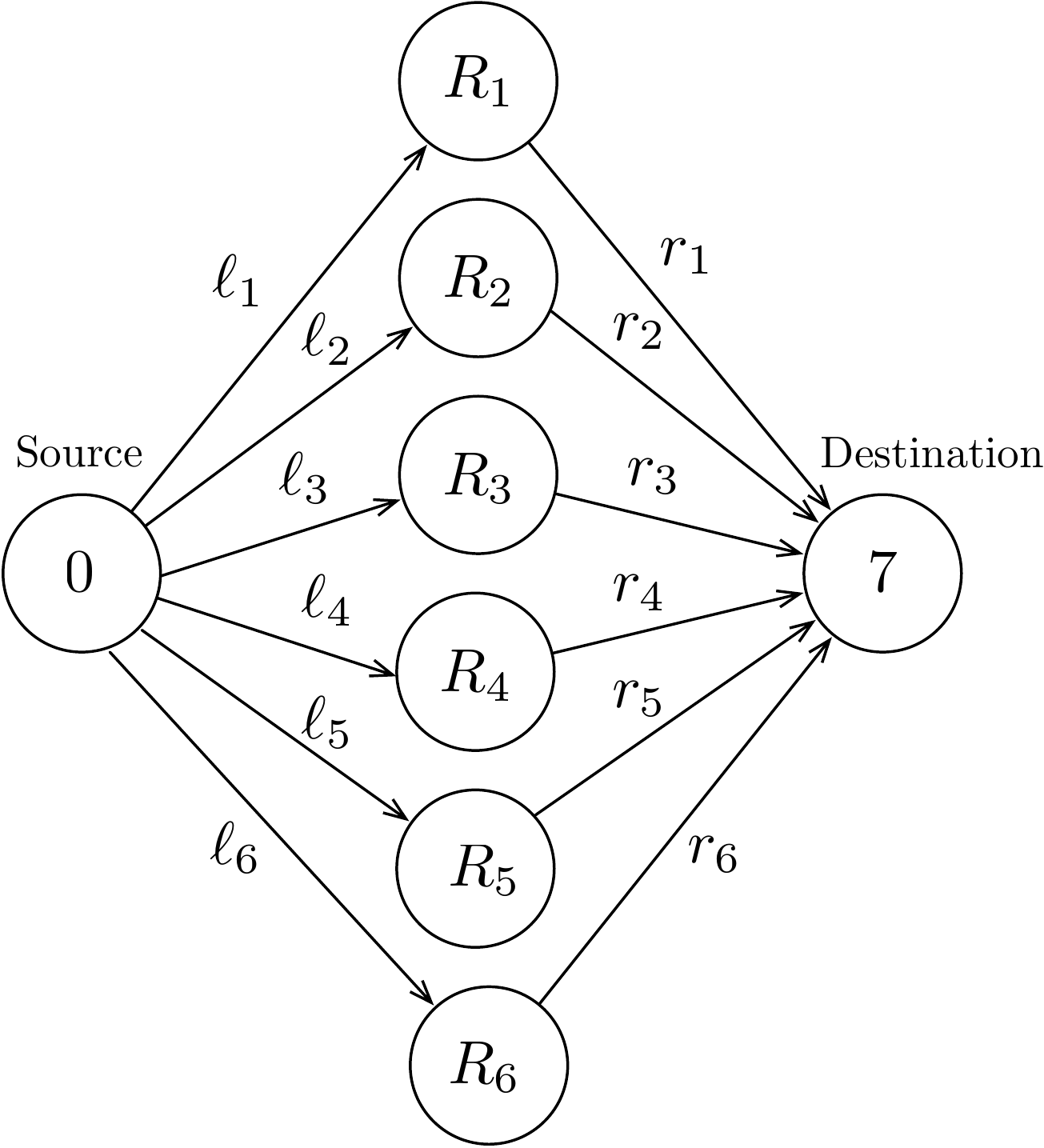}}}
	\hspace{5mm}
	\subfloat[\label{fig:even-1}]{	
	\adjustbox{valign=c}{\footnotesize{
	\begin{tabular}[h]{c||c|c}
    $i$ & $\ell_i$   & $r_i$\\
    \hline
    $1$ & $\sqrt{2}$ & $2+\sqrt{2}$\\
    $2$ & $\sqrt{2}$ & $2+\sqrt{2}$\\
    $3$ & $2$		 & $2$\\ 
    $4$ & $2$		 & $2$\\
    $5$ & $2+\sqrt{2}$ & $\sqrt{2}$\\
    $6$ & $2+\sqrt{2}$ & $\sqrt{2}$
    \end{tabular}  
	}}}
	\hspace{5mm}
	\subfloat[\label{fig:even-2}]{
	\adjustbox{valign=c}{
	\footnotesize{
    \begin{tabular}[h]{c||c|c}
    $i$ & $\ell_i$   & $r_i$\\
    \hline 
     $1$ & $1$ & $L\rightarrow \infty$\\
     $2$ & $\frac{2+\sqrt{2}}{2}$ & $1+\sqrt{2}$\\
     $3$ & $\frac{2+\sqrt{2}}{2}$ & $1+\sqrt{2}$\\
     $4$ & $1+\sqrt{2}$ & $\frac{2+\sqrt{2}}{2}$\\
     $5$ & $1+\sqrt{2}$ & $\frac{2+\sqrt{2}}{2}$\\
	$6$ & $L\rightarrow \infty$ & $1$\\
    \end{tabular}
   	 }}
	}
	\caption{Gaussian half-duplex diamond networks with $n=6$ relays for which the bound in~\eqref{ratio} is tight. The table in (b) shows the link capacities for the network defined  in \eqref{eq:WN1} and the table in (c) indicates the link capacities of the network given in \eqref{eq:WN1-2}.}
	\label{fig:WorstNetwork1}
	\vspace{-5mm}
\end{figure}

\noindent {\textbf{Case B.1: }} Let $n=2k+1$ be an odd number. We consider a Gaussian half-duplex diamond $n$-relay network $\mathcal{N}$ for which
\begin{align}
\label{eq:WN2}
\begin{split}
\ell_1 & = 1,
\qquad
r_1  = L \to \infty,
\\
\ell_{2i} & = \ell_{2i+1} = \frac{\sin \left ( i \theta \right ) + \sin \left ( (i+1)\theta\right )}{\sin \left ( (i+1) \theta \right )}, \qquad i \in \left [1:k \right],
\\
r_{2i} & = r_{2i+1} = \frac{\sin \left ( i \theta \right ) + \sin \left ( (i+1)\theta\right )}{\sin \left ( i \theta \right )}, \qquad i \in \left [1:k \right],
\\
\theta & = \frac{2 \pi}{n+2}.
\end{split}
\end{align}

Similar to Case A.1,  the network in~\eqref{eq:WN2} satisfies  $\ell_1 \leq \ell_2 \leq \ldots \leq \ell_n$ and
$r_1 \geq r_2 \geq \ldots \geq r_n$. Moreover, 
the single relay approximate capacities satisfy
\begin{align}
\label{eq:sinRelWN2}
\mathsf{C}_1(\mathcal{N}_i) = \frac{\ell_i r_i}{\ell_i + r_i}  = \left(\frac{1}{\ell_i} + \frac{1}{r_i}\right)^{-1} = 1,
\end{align}
for $i \in [1:n]$, which implies $\mathsf{C}_1(\mathcal{N})=1$, i.e., the best relay in $\mathcal{N}$ has unitary approximate capacity. Furthermore, for any $t\in[1:n]$ with $i=\lfloor t/2\rfloor$ we have 
\begin{align*}
\ell_t + r_{t+2} &= \frac{\sin \left ( i \theta \right ) + \sin \left ( (i+1)\theta\right )}{\sin \left ( (i+1) \theta \right )} + \frac{\sin \left ( (i+1) \theta \right ) + \sin \left ( (i+2)\theta\right )}{\sin \left ( (i+1) \theta \right )}\nonumber\\
&= \frac{2\sin \left ( (i+1) \theta \right ) + 2\sin \left ( (i+1)\theta\right ) \cos (\theta)}{\sin \left ( (i+1) \theta \right )} = 2\cos(\theta)+2.
\end{align*}
where we let $r_{n+1}=r_{n+2}=0$.
Therefore, similar to \eqref{eq:CLB} we can show that $\mathsf{R}_n(\mathcal{N}) = 2\cos(\theta)+2$ is achievable for this network, using the two-state schedule given by 
\begin{align*}
\lambda_{\cS} = \left\{
\begin{array}{ll}
\frac{1}{2} & \textrm{if $\cS = \cS_o = \{3,5,\dots, 2k+1\}$},\\
\frac{1}{2} & \textrm{if $\cS = \cS_e = \{2,4,6,\dots, 2k\}$},\\
0 & \textrm{otherwise.}
\end{array}
\right.
\end{align*}
Note that in this schedule, relay $R_1$ is (asymptotically) always receiving, since its transmit capacity is unboundedly greater than its receive capacity. Moreover, similar to \eqref{eq:CUB}, we can argue that $\mathsf{C}_n(\mathcal{N}) \leq \ell_n = 2\cos(\theta)+2$. Therefore, we get 
\begin{align*} 
\frac{\mathsf{C}_{1}(\mathcal{N})}{{\mathsf{C}}_{n}(\mathcal{N})} = \frac{1}{\cos \left (\frac{2\pi}{n+2} \right )+2},
\end{align*}
which proves the tightness of the bound in~\eqref{ratio} when $n$ is odd. Note that this network topology corresponds to Case~III of the network analysis in Appendix~\ref{app:4cases}.
An example of the network construction in~\eqref{eq:WN2} for $n=5$ is provided in Fig.~\ref{fig:odd-1}.

\noindent {\textbf{Case B.2: }} The second network configuration that satisfies the bound in~\eqref{ratio} with equality for an odd number of relays, i.e., $n=2k+1$, is given by
\begin{align}
\label{eq:WN2-2}
\begin{split}
\ell_{2i-1} =\ell_{2i}&= 
%\frac{\cos \left ( (i-1) \theta \right ) - \cos \left ((i+1) \theta \right )}{\cos \left ( i \theta \right ) - \cos \left ( (i+1) \theta \right )}= 
\frac{2\sin(\theta) \sin \left ( i \theta \right ) }{\cos \left ( i \theta \right ) - \cos \left ( (i+1) \theta \right )}, \qquad i \in [1:k],
\\
r_{2i-1} =r_{2i} &= 
%\frac{\cos \left ( (i-1) \theta \right ) - \cos \left ((i+1) \theta \right )}{\cos \left ( (i-1) \theta \right )-\cos \left ( i \theta \right ) } = 
\frac{2\sin(\theta) \sin \left ( i \theta \right )}{\cos \left ( (i-1) \theta \right )-\cos \left ( i \theta \right ) }, \qquad i \in [1:k],
\\
\ell_n&=L\rightarrow \infty, \qquad
r_n=1,\\
\theta & = \frac{2 \pi}{n+2}.
\end{split}
\end{align}
It is easy to see that this network also satisfies $\ell_1 \leq \ell_2 \leq \ldots \leq \ell_n$ and
$r_1 \geq r_2 \geq \ldots \geq r_n$.
%the increasing order on the left side and the decreasing order on the right side. 
Moreover, the approximate single relay capacities equal one, and hence $\mathsf{C}_1(\mathcal{N})=1$. Furthermore, the approximate capacity of the entire network is $\mathsf{C}_n(\mathcal{N}) = 2\cos(\theta)+2$, which can be achieved  using the two-state schedule given by 
\begin{align*}
\lambda_{\cS} = \left\{
\begin{array}{ll}
\frac{1}{2} & \textrm{if $\cS = \cS_o = \{1,3,5,\dots, 2k+1\}$},\\
\frac{1}{2} & \textrm{if $\cS = \cS_e = \{2,4,6,\dots, 2k, 2k+1\}$},\\
0 & \textrm{otherwise,}
\end{array}
\right.
\end{align*}
i.e., the relay node $R_n$ is always in transmit mode. This leads to 
\begin{align*} 
\frac{\mathsf{C}_{1}(\mathcal{N})}{{\mathsf{C}}_{n}(\mathcal{N})} = \frac{1}{\cos \left (\frac{2\pi}{n+2} \right )+2},
\end{align*}
which shows that the network in \eqref{eq:WN2-2} satisfies the bound in~\eqref{ratio} with equality.  Note that this network topology corresponds to Case~II of the network analysis in Appendix~\ref{app:4cases}. An example of such network for $n=5$ relay nodes is shown in Fig.~\ref{fig:odd-2}. It is worth noting that the two network topologies introduced for an odd number of relays are indeed identical up to flipping of the left and right point-to-point link capacities, and relabeling of the relays.

\begin{figure}
	\centering
	\subfloat[]{\adjustbox{valign=c}{\includegraphics[width = 0.35\textwidth]{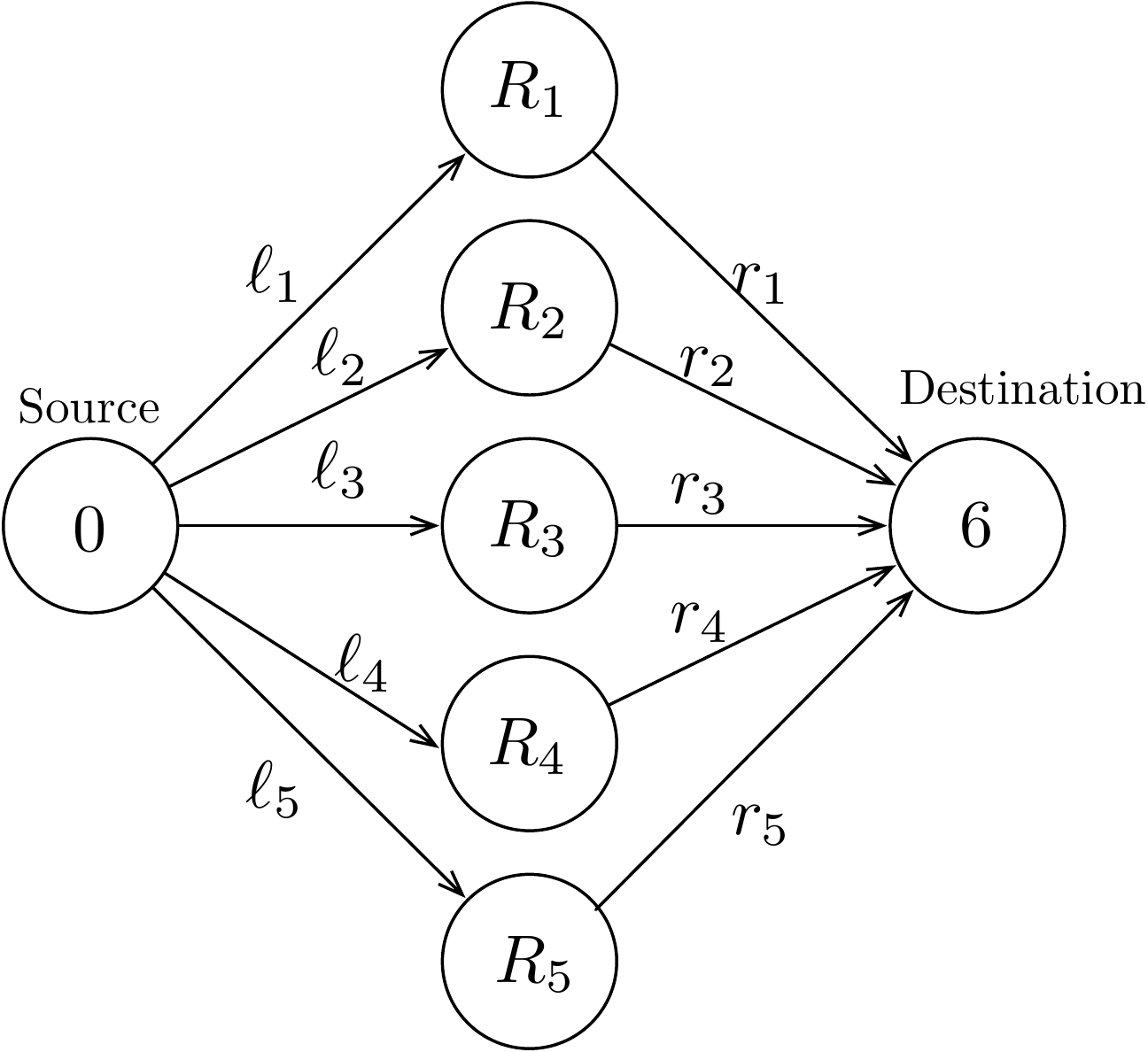}}}
	\hspace{5mm}
	\subfloat[\label{fig:odd-1}]{	
	\adjustbox{valign=c}{\footnotesize{
	\begin{tabular}[h]{c||c|c}
    $i$ & $\ell_i$   & $r_i$\\
    \hline
    $1$ & $1$ & $L\rightarrow \infty$\\
    $2$ & $1.8019$ & $2.2470$\\
    $3$ & $1.8019$ & $2.2470$\\
    $4$ & $3.2470$ & $1.4450$\\
    $5$ & $3.2470$ & $1.4450$\\
    \end{tabular}  
	}}}
	\hspace{5mm}
	\subfloat[\label{fig:odd-2}]{
	\adjustbox{valign=c}{
	\footnotesize{
    \begin{tabular}[h]{c||c|c}
    $i$ & $\ell_i$   & $r_i$\\
    \hline 
     $1$ & $1.4450$ & $3.2470$\\
     $2$ & $1.4450$ & $3.2470$\\     
     $3$ & $2.2470$ & $1.8019$\\
     $4$ & $2.2470$ & $1.8019$\\
     $5$ & $L\rightarrow \infty$ & $1$\\
    \end{tabular}
   	 }}
	}
	\caption{Gaussian half-duplex diamond networks with $n=5$ relays for which the bound in~\eqref{ratio} is tight. The table in (b) shows the link capacities of the network given in \eqref{eq:WN2} and the table in (c) indicates the link capacities of the network in \eqref{eq:WN2-2}.}
	\label{fig:WorstNetwork2}
	\vspace{-5mm}
\end{figure}

%
%\section{Conclusion}
%\label{sec:concl}

\appendices

\section{Proof of Proposition~\ref{lem:Smn}}
\label{app:4cases}
We consider the four possible cases, depending on the values of $\beta_1^\star$ and $\beta_m^\star$.

\noindent 
{\textbf{Case I: } $\beta_1^\star>0$ and $\beta_m^\star<\infty$.} 
Since $\beta_1^\star>0$, then from Lemma~\ref{lemm:OptSol}, we know that
\begin{align*}
1+ \frac{1}{\beta_1^\star} = G_0(\bm{\beta}^\star) = \mathrm{OPT}_5(m)\stackrel{\eqref{eq:sn}}{=} \sigma_{n,m} +2 \Rightarrow \frac{1}{\beta_1^\star} = \sigma_{n,m} +1.
\end{align*}
Moreover, using~\eqref{eq:B1g0} inside~\eqref{b_i}, we obtain
\begin{equation}
\label{eq:IS1}
\begin{split}
\left \{
\begin{array}{l}
uU^0+vV^0=b_0=1, \\
uU^1+vV^1=b_1=\frac{1}{\beta_1^\star} = \sigma_{n,m} +1,
\end{array}
\right .
\quad
\Rightarrow \quad 
\left \{
\begin{array}{l}
%u=\frac{U-1}{\sigma-2}, \\
%v=\frac{V-1}{\sigma-2}.
 u = \frac{U-1}{\sigma_{n,m}-2}, \\ 
v = \frac{V-1}{\sigma_{n,m}-2}.
\end{array}
\right .
\end{split}
\end{equation}
Then, since $\beta_m^\star<\infty$,  Lemma~\ref{lemm:OptSol} implies that 
\begin{equation*}   
    1+ \beta_m^\star = G_m(\bm{\beta}^\star) = \mathrm{OPT}_5(m) =   2+ \sigma_{n,m}, 
    \end{equation*}
    or equivalently, 
    \begin{equation*}
    \sigma_{n,m} + 1 = \beta_m^\star 
    =\frac{b_{m-1}}{b_m}=\frac{uU^{m-1}+vV^{m-1}}{uU^{m}+vV^m}.
\end{equation*}
Therefore, 
\begin{align}
    0&=u \left( U^m(\sigma_{n,m}+1)-U^{m-1}\right)+v \left(V^m(\sigma_{n,m}+1)-V^{m-1}\right) \nonumber
    \\&=uU^m(U+1)+vV^m(V+1), \label{eq:IS}
\end{align}
where the last equality follows since we have
\begin{align*}
 U^m(\sigma_{n,m}+1)-U^{m-1} =  U^{m-1} (U \sigma_{n,m} +U -1) \stackrel{{\rm{(a)}}}{=}  U^{m-1} (U^2 +U ) = U^m (U+1) ,
\end{align*}
and the equality in $\rm{(a)}$ follows from the characteristic function in~\eqref{eq:charfun}.
Therefore, since $UV=1$, from~\eqref{eq:IS} we obtain
\begin{equation*}
    U^{2m}=\left (\frac{U}{V}\right )^m=-\frac{v}{u}\frac{V+1}{U+1}\stackrel{{\rm{(b)}}}{=}-\frac{V-1}{U-1} \frac{V+1}{U+1}\stackrel{{\rm{(c)}}}{=}\frac{1}{U^2} \Rightarrow  U^{2m+2}=1,
\end{equation*}
where the equality in $\rm{(b)}$ follows by using the values in~\eqref{eq:IS1} for $u$ and $v$, and the equality in $\rm{(c)}$ follows by substituting $V=1/U$. Thus, we get $2m+2$ pairs of $(U,V)$, enumerated by a parameter $k\in [0:2m+1]$, given by 
\begin{equation*}
    U(k)=\exp \left (
    \frac{2k\pi i}{2m+2}\right ),
    \qquad
        V(k)=\exp \left (-
    \frac{2k\pi i}{2m+2}\right ).
\end{equation*}
Therefore, we have
\begin{equation*}
\sigma_{n,m}(k)=U(k)+V(k)=\exp \left (\frac{2k\pi j}{2m+2} \right )+\exp \left (-\frac{2k\pi j}{2m+2}\right )=2 \cos \left (\frac{2k \pi}{2m+2} \right ).
\end{equation*}
Note that $\sigma_{n,m}$ above is a function of $k$.
However, the choice of $k=0$ leads to $U=V=1$ and $\sigma_{n,m}=2$ which is an invalid choice (see Footnote~\ref{ftnt:sigma2}). Other than that, for every given $m$ we have
\begin{align*}
\sigma_{n,m} = \max_{\substack{k\in [0:2m+1]\\ k\neq 0}} \sigma_{n,m}(k) = \sigma_{n,m}(1) = 2 \cos \left (\frac{2 \pi}{2m+2} \right ),
\end{align*}
which proves our claim in Proposition~\ref{lem:Smn} when $\beta_1^\star>0$ and $\beta_m^\star<\infty$.
%Finally, we also optimize $\sigma_n$ over $m$. 

\noindent{\textbf{Case II: } $\beta_1^\star>0$ and $\beta_m^\star=\infty$.} The initial condition of the recurrence relation are identical to that of Case I. Hence, we get $b_i = u U^i + v V^i$, where $u$ and $v$ are given in \eqref{eq:IS1}. 
Moreover, $\beta_m^\star=\infty$ implies $b_m=0$. Substituting this in~\eqref{b_i} for $i=m$ leads to
\begin{equation*}
    0=b_m=uU^m+vV^m,
\end{equation*}
which implies
\begin{equation*}
    U^{2m}\stackrel{{\rm{(a)}}}{=} \left (\frac{U}{V} \right )^m=-\frac{v}{u}\stackrel{{\rm{(b)}}}{=}-\frac{V-1}{U-1}\stackrel{{\rm{(a)}}}{=}\frac{1}{U} \Rightarrow U^{2m+1} =1,
\end{equation*}
where the equalities in $\rm{(a)}$ are due to the fact that $V=1/U$, and that in $\rm{(b)}$ follows from~\eqref{eq:IS1}. Thus,
\begin{align*}
U(k)=\exp{\left (\frac{2k\pi j}{2m+1}\right)}, \qquad  V(k)=\exp{\left (-\frac{2k\pi j}{2m+1}\right)},
\end{align*}
and hence,
\begin{equation*}
    \sigma_{n,m}(k)=U(k)+V(k)=\exp{\left(\frac{2k\pi j}{2m+1}\right)}+\exp{\left(-\frac{2k\pi j}{2m+1}\right)}=2\cos{\left(\frac{2k\pi }{2m+1}\right)},
\end{equation*}
for $k \in [0:2m]$. 
Maximizing $\sigma_{n,m}(k)$ we get 
\begin{equation*}
    \sigma_{n,m}=\max_{\substack{k \in [0:2m] \\ k\neq 0}} \sigma_{n,m}(k)= \sigma_{n,m}(1) = 2\cos{\left(\frac{2\pi }{2m+1}\right)},
\end{equation*}
as claimed in Proposition~\ref{lem:Smn}.

\noindent{\textbf{Case III: } $\beta_1^\star=0$ and $\beta_m^\star<\infty$.} When $\beta_1^\star=0$, the initial conditions of the recurrence  equation are given in~\eqref{eq:B1e0}. We have
\begin{align}
\label{eq:ISC3}
\begin{split}
\left \{
\begin{array}{l}
uU^0+vV^0=b_0=0, \\
uU^1+vV^1=b_1=1, 
\end{array}
\right.
\quad
\Rightarrow \quad 
\left \{
\begin{array}{l}
u=\frac{1}{\sqrt{\sigma_{n,m}^2-4}}, \\
v=-\frac{1}{\sqrt{\sigma_{n,m}^2-4}} .
\end{array}
\right .
\end{split}
\end{align}
Moreover, Lemma~\ref{lemm:OptSol} for $\beta_m^\star<\infty$ implies 
\begin{align*}
1+ \beta_m^\star = G_m(\bm{\beta}^\star) = \mathrm{OPT}_5(m) = 2+\sigma_{n,m} \quad \Rightarrow \quad 1+ \sigma_{n,m}=\beta_m^\star=\frac{b_{m-1}}{b_m}.
\end{align*}
Hence, 
\begin{align*}
uU^{m-1} + v V^{m-1} = b_{m-1} = (1+ \sigma_{n,m}) b_{m} = (1+ \sigma_{n,m})\left(uU^{m} + v V^{m} \right)
\end{align*}
or equivalently,
\begin{align*}
uU^{m-1} (U+ \sigma_{n,m} U -1 )+ v V^{m-1} (V+ \sigma_{n,m} V -1 ) 
\stackrel{{\rm{(a)}}}{=} uU^{m-1} (U+ U^2 )+ v V^{m-1} (V+ V^2 ) 
= 0,   
\end{align*}
where $\rm{(a)}$ follows from the fact that $U$ and $V$ are the roots of the characteristic function in~\eqref{eq:charfun}. 
Therefore,  we get
\begin{equation*}
    U^{2m}=\left(\frac{U}{V}\right )^m=-\frac{v}{u} \frac{V+1}{U+1} \stackrel{{\rm{(a)}}}{=} \frac{1}{U} \Rightarrow U^{2m+1}=1,
\end{equation*}
where the equality in $\rm{(a)}$ follows from~\eqref{eq:ISC3}. Therefore, similar to Case~II, we get
\begin{equation*}
    \sigma_{n,m}=2 \cos{\left(\frac{2\pi}{2m+1}\right)},
\end{equation*}
which proves our claim in Proposition~\ref{lem:Smn}.

\noindent{\textbf{Case IV: } $\beta_1^\star=0$ and $\beta_m^\star=\infty$.} Since $\beta_1^\star=0$ , the initial conditions of this case are identical to those of Case~III given in~\ref{eq:ISC3}.  However, from $\beta_m^\star=\infty$ we have $b_m=0$, which implies
\begin{equation*}
    0=b_m=uU^m+v V^m.
\end{equation*}
This leads to
\begin{equation}
    U^{2m}=\left(\frac{U}{V}\right)^m=-\frac{v}{u}\stackrel{{\rm{(a)}}}{=}1 \Rightarrow U^{2m}=1,
\end{equation}
where the equality in $\rm{(a)}$ follows by using~\eqref{eq:ISC3}. Thus,
\begin{align*}
U(k)=\exp{\left(\frac{2k\pi j}{2m}\right)}, \qquad V(k)=\exp{\left(-\frac{2k\pi j}{2m}\right)},
\end{align*}
and
\begin{equation*}
    \sigma_{n,m}(k)=U(k)+V(k)= 2 \cos{\left(\frac{2k\pi}{2m}\right)},
\end{equation*}
for some $k \in [0:2m-1]$. Maximizing $\sigma_{n,m}(k)$ over $k\neq 0$ we get 
\begin{equation}
    \sigma_{n,m}= \max_{\substack{k \in [0:2m-1] \\ k\neq 0}} \sigma_{n,m}(k)= \sigma_{n,m}(1)  = 2 \cos{\left(\frac{2\pi}{2m}\right)}.
\end{equation}
This proves our claim in Proposition~\ref{lem:Smn}, for the forth case when $\beta_1^\star=0$ and $\beta_m^\star = \infty$.
%The proof of Lemma~\ref{lem:Smn} is therefore concluded.

 \bibliography{ReviewBibSJ1.bib}
 \bibliographystyle{IEEEtran}

\end{document}